\documentclass[10pt,journal,compsoc]{IEEEtran}

\usepackage{tikz}
\usetikzlibrary{decorations.pathreplacing}
\usetikzlibrary{math} 
\usetikzlibrary{tikzmark}
\usetikzlibrary{calc}

\usepackage{framed,enumitem}
\newlist{todolist}{itemize}{2}
\setlist[todolist]{label=$\square$}
\usepackage{booktabs} 
\usepackage[utf8]{inputenc}
\usepackage{tcolorbox}

\usepackage{listings}
\usepackage{longtable}
\usepackage{multirow}
\usepackage{graphicx}
\usepackage{subcaption}
\usepackage[noend]{algpseudocode}
\algrenewcommand\textproc{}
\algdef{SE}[EVENT]{Event}{EndEvent}[1]{\textbf{upon event}\ #1\ \algorithmicdo}{\algorithmicend\ \textbf{event}}%
\algtext*{EndEvent}

\usepackage{soul}

\newif\iftodo
\todotrue

\newif\ifbreakpage

\newcommand{\separate}{\ifbreakpage \clearpage \fi}

\usepackage{amsthm}
\theoremstyle{definition}
\newtheorem{definition}{Definition}
\theoremstyle{observation}
\newtheorem{observation}{Observation}
\theoremstyle{lemma}
\newtheorem{lemma}{Lemma}
\theoremstyle{theorem}
\newtheorem{theorem}{Theorem}
\theoremstyle{corollary}
\newtheorem{corollary}{Corollary}
\theoremstyle{remark}

\makeatletter
\renewenvironment{proof}[1][\proofname] {\par\pushQED{\qed}\normalfont\topsep6\p@\@plus6\p@\relax\trivlist\item[\hskip\labelsep\bfseries#1\@addpunct{.}]\ignorespaces}{\popQED\endtrivlist\@endpefalse}
\makeatother

\usepackage{caption}
\newcommand{\xxx}[1]{\textit{STRETCH}}

\makeatletter
\lst@InstallKeywords k{attributes}{attributestyle}\slshape{attributestyle}{}ld
\makeatother

\lstdefinestyle{mystyle}{
  commentstyle=\color{codegreen},
  keywordstyle=\color{magenta},
  numberstyle=\tiny\color{codegray},
  stringstyle=\color{codepurple},
  basicstyle=\footnotesize,
  breakatwhitespace=false,   
  frame=single,
  breaklines=true,                 
  captionpos=b,                    
  keepspaces=true,                 
  numbers=left,                    
  numbersep=5pt,                  
  showspaces=false,                
  showstringspaces=false,
  showtabs=false,
  tabsize=1,
  moreattributes={fromR,poll,peek,predicate,add,join,getNextReady,getStatesIndexes,get,process,initialize,hash},
  attributestyle = \bfseries\color{Blue},
  literate={\ \ }{{\ }}1
}

\ifCLASSOPTIONcompsoc
  \usepackage[nocompress]{cite}
\else
  \usepackage{cite}
\fi

\usepackage[normalem]{ulem}

\usepackage{amssymb,amsmath,amsthm}
\usepackage{float}

\newcommand{\draft}[1]{{\color{cyan} #1}}
\renewcommand{\draft}[1]{#1}

\newcommand{\scalegate}{$SG$}
\newcommand{\elasticscalegate}{$ESG$}
\newcommand{\elasticscalegatein}{$ESG_{in}$}
\newcommand{\elasticscalegateout}{$ESG_{out}$}

\newcommand{\WA}{\ensuremath{\textit{WA}}}
\newcommand{\WS}{\ensuremath{\textit{WS}}}
\newcommand{\I}{\ensuremath{\textit{I}}}
\newcommand{\psipar}{\ensuremath{\textit{WT}}}

\newcommand{\keybysingle}{\ensuremath{f_{\textit{SK}}}}
\newcommand{\keybysingleinmath}{\keybysingle}

\newcommand{\keybymulti}{\ensuremath{f_{\textit{MK}}}}
\newcommand{\keybymultiinmath}{\keybymulti}

\newcommand{\mapop}{$M$}
\newcommand{\aggop}{$A$}
\newcommand{\joinop}{$J$}

\newcommand{\aggplusop}{A^+}
\newcommand{\aggplusopinmath}{$\aggplusop$}

\newcommand{\joinplusop}{J^+}
\newcommand{\joinplusopinmath}{$\joinplusop$}

\newcommand{\generalopplus}{O^+}
\newcommand{\generalopplusinmath}{$\generalopplus$}

\newcommand{\generalopplusinst}{o^+}
\newcommand{\generalopplusinstinmath}{$\generalopplusinst$}

\newcommand{\opstate}{\sigma}
\newcommand{\opstateinmath}{$\opstate$}

\newcommand{\opmappingfun}{f_{\mu}}
\newcommand{\opmappingfuninmath}{$\opmappingfun$}
\newcommand{\opmapping}{f_{\mu}}
\newcommand{\opmappinginmath}{$\opmappingfun$}

\newcommand{\watermarkof}[1]{W_{#1}^\omega}
\newcommand{\watermarkofinmath}[1]{$\watermarkof{#1}$}

\newcommand{\X}{\ensuremath{\textit{TB}}}
\newcommand{\Xinmath}{$\X$}
\newcommand{\Xplus}{\X}

\usepackage[ruled,vlined,linesnumbered,norelsize,noend]{algorithm2e}

\SetCommentSty{mycommfont}
\let\oldnl\nl
\newcommand{\nonl}{\renewcommand{\nl}{\let\nl\oldnl}}

\usepackage{hyperref}

\usepackage{array}
\usepackage{bm}

\renewenvironment{description}[1][0pt]
  {\list{}{\labelwidth=0pt \leftmargin=#1
   }}
  {\endlist}

\usepackage{threeparttable}
\usepackage{mathtools}
\usepackage{tabularx}

\usepackage{booktabs}

\usepackage{color, colortbl}
\definecolor{Gray}{gray}{0.9}

\begin{document}
\title{STRETCH: Virtual Shared-Nothing Parallelism for Scalable and Elastic Stream Processing}

\author{Vincenzo~Gulisano, Hannaneh~Najdataei, Yiannis~Nikolakopoulos, Alessandro~V.~Papadopoulos, Marina~Papatriantafilou, and~Philippas~Tsigas
\IEEEcompsocitemizethanks{
\IEEEcompsocthanksitem A preliminary version of this paper appeared at the ACM International Conference on Distributed and Event-based Systems (DEBS'19)~\cite{najdataei2019stretch}.
\IEEEcompsocthanksitem V. Gulisano, H. Najdataei, M. Papatriantafilou and P. Tsigas are with Chalmers University of Technology, G{\"o}teborg; Y. Nikolakopoulos is with ZeroPoint Technologies and contributed to this work while he was with Chalmers University of Technology; A. V. Papadopoulos is with M{\"a}lardalen University, V{\"a}ster{\aa}s. E-mail: \{vincenzo.gulisano,hannajd,ptrianta,tsigas\}@chalmers.se; yiannis@zptcorp.com; alessandro.papadopoulos@mdh.se;}
}

\IEEEtitleabstractindextext{%
\begin{abstract}
Stream processing applications extract value from raw data through Directed Acyclic Graphs of data analysis tasks.
Shared-nothing (SN) parallelism is the de-facto standard to scale stream processing applications.
Given an application, SN parallelism instantiates several copies of each analysis task, making each instance responsible for a dedicated portion of the overall analysis, and relies on dedicated queues to exchange data among connected instances.
On the one hand, SN parallelism can scale the execution of applications both up and out since threads can run task instances within and across processes/nodes.
On the other hand, its lack of sharing can cause unnecessary overheads and hinder the scaling up when threads operate on data that could be jointly accessed in shared memory.
This trade-off motivated us in studying a way for stream processing applications to leverage shared memory and boost the scale up (before the scale out) while adhering to the widely-adopted and SN-based APIs for stream processing applications.

We introduce \xxx{}, a framework that maximizes the scale up and offers instantaneous elastic reconfigurations (without state transfer) for stream processing applications. We propose the concept of Virtual Shared-Nothing (VSN) parallelism and elasticity and provide formal definitions and correctness proofs for the semantics of the analysis tasks supported by \xxx{}, showing they extend the ones found in common Stream Processing Engines.
\draft{We also provide a fully implemented prototype and show that \xxx{}'s performance exceeds that of state-of-the-art frameworks such as Apache Flink and offers, to the best of our knowledge, unprecedented ultra-fast reconfigurations, taking less than 40 ms even when provisioning tens of new task instances.}
\end{abstract}

}

\maketitle

\separate
\section{Introduction}\label{sec:introduction}

Stream processing applications process data (\textit{tuples}) coming from \textit{unbounded streams} (e.g., tweets or trade records).
Such applications are run by Stream Processing Engines (SPEs) such as Apache Flink~\cite{flink} or Storm~\cite{storm}.
SPEs provide users with semantically-rich operators composable into Directed Acyclic Graphs (DAGs) and automate the process of deploying and executing efficiently user-defined DAGs.

SPEs' de-facto standard to parallelize the execution of stream processing applications builds on \textit{Shared-Nothing} (SN) key-by parallelism.
Simply put, the idea is to create multiple instances of each operator, each with a dedicated state and queues to exchange and route each tuple to exactly one of such instances.
SN parallelism is also used in elasticity protocols, which aim at adjusting the parallelism degree of the DAGs run by SPEs, avoiding the overheads of over- or under-provisioned applications~\cite{gulisano2012streamcloud}. 

\subsubsection*{Trade-offs of shared-nothing parallelism.}
While able to scale the execution of a DAG both up and out SN parallelism can incur unnecessary overheads.
The first type of overhead is caused by operators' dedicated input/output queues.
When the semantics of a parallel operator $O$ require to route a tuple to multiple instances of $O$, this leads to data duplication.
Consider an application, which we use as a running example, in which $O$ computes the longest tweet on a per-hour, per-hashtag basis. 
Its analysis can be parallelized by having each parallel instance of $O$ responsible for one observed hashtag.
Tweets carrying multiple hashtags, though, might need to be shared with multiple instances.
This overhead is exacerbated by the fact that SPEs might leave to users the task of correctly customizing the routing of tuples.
The second type of overhead is caused by the operators' dedicated internal state.
Since instances' states are not shared, state transfer is needed in elastic reconfigurations to adjust the workload distribution and/or parallelism degree of an operator~\cite{gulisano2012streamcloud}.
This overhead is exacerbated when SPEs request users to implement serialization/deserialization methods for custom states~\cite{flinkstate}.

While these overheads are unavoidable for operator instances running distributedly, they are not for instances that share memory within the same process, and could thus share tuples and states too.
Following the principle that applications should be properly scaled up before being scaled out~\cite{gibbons2015big},
we thus pose the following question: 
\textit{How can we seamlessly leverage shared memory to boost parallel/elastic executions of common SPEs' operators, while avoiding data duplication and state transfer overheads?}

\subsubsection*{Contributions}
\draft{We formally show that it is possible to define parallel and elastic SPE operators that, by \textit{virtualizing} the common Application Programming Interfaces (APIs) based on SN parallelism, can leverage shared memory to first scale streaming applications up while allowing to rely on SN parallelism to later scale them out.}
We thus introduce \xxx{} and the concept of \emph{Virtual Shared-Nothing} (VSN) parallelism/elasticity in stream processing. 
These are our contributions:
\begin{itemize}[leftmargin=*]
    \item we prove the potential need for data duplication in SN parallelism and propose a unified generalized model for SN parallelism encapsulating common SPEs' operators,
    \item we prove that VSN parallelism can correctly enforce the semantics of our generalized model while circumventing the overheads of data duplication and state transfer,
    \item we provide a fully implemented prototype, which builds on state-of-the-art data structures such as ScaleGate~\cite{cederman2013concurrent} and our extended \emph{Elastic ScaleGate} implementation, as well as  a thorough evaluation with several use-cases and both synthetic and real data.
\end{itemize}

\noindent\textit{Outline:} 
\autoref{sec:prel} covers preliminaries, 
\autoref{sec:smps} formalizes our problem,
\autoref{sec:generalizedmodel} discusses the data duplication overhead and introduces our generalized model,
\autoref{sec:stretch}-\autoref{sec:implementation} cover \xxx{}'s model and implementation,
\autoref{sec:eval} evaluates \xxx{}, \autoref{sec:rw} discusses related work, and \autoref{sec:conc} wraps up our work.
A table of abbreviations/symbols is found in Appendix~\ref{apx:symbols}.

\separate
\section{Preliminaries}
\label{sec:prel}

\subsection{Stream processing basics}\label{sec:dsbasics}

In accordance with the DataFlow model~\cite{akidau2015dataflow},
\textit{streams} are unbounded sequences of \textit{tuples}.
Tuples have two \textit{attributes}: the metadata and the payload $\varphi$.
The metadata carries the timestamp $\tau$ and possibly further \textit{sub-attributes}. 
We write $t.\tau$ to refer to the $\tau$ sub-attribute of tuple $t$'s metadata.
We reference $\varphi$'s $\ell$-th sub-attribute as $t.\varphi[\ell]$.
A tuple's combined notation is $\langle \tau,\dots, \left[ \varphi[1], \varphi[2], \dots \right] \rangle$.

\textit{Stream processing queries} (or simply queries) are composed of \textit{ingresses}, \textit{operators}, and \textit{egresses}.
Ingresses forward \textit{ingress tuples} 
(e.g., events reported by sensors or other applications).
Each \textit{ingress stream} can be fed to 
one or more  operators, the basic units manipulating tuples. 
Operators, connected 
in a DAG, process input tuples and produce output tuples;
eventually, \textit{egress tuples} are fed to egresses, 
which deliver results to end-users or other applications.

As ingress tuples correspond to events, $\tau$ is the \emph{event time} set by the ingress to when the event took place.
Operators set $\tau$ of each output tuple according to their semantics, while $\varphi$ is set by user-defined functions.
Event time is expressed in time units from a given epoch, 
and progresses in SPE-specific discrete $\delta$ increments (e.g., milliseconds~\cite{flink}). 

The operators of a query are either \textit{stateless} or \textit{stateful}.
Stateless operators process each tuple individually. The Map/Flatmap (\mapop{}) operator, for instance, transforms each input tuple $t_{in}$ into one or more output tuples $t_{out}$ by setting $t_{out}.\tau$ to $t_{in}.\tau$ (we write $t_{out}.\tau \xleftarrow{} t_{in}.\tau$) and using a user-defined function to create $t_{out}.\varphi$ from $t_{in}.\varphi$.
Stateful operators run their analysis on delimited 
groups of tuples called \textit{windows}, as explained next.

\subsubsection*{Stateful analysis over time-based windows}
We focus on common stateful operators running their analysis over \textit{time-based windows}, namely \textit{Aggregates} and \textit{Joins}. 
For conciseness, for operator $O$ we use the notation
$$O(\WA,\WS,\I,\keybysingle,\psipar,S,f_1,f_2,\ldots)$$ 
to refer to the parameters that such stateful operators share. 
\begin{description}
\item [\draft{Parameters Window Advance (\WA{}) and Size (\WS{})}] define the advance/size of $O$'s windows, respectively, which cover periods $[\ell WA, \ell WA+ WS)$, with $\ell \in \mathbb{Z}$.
Consecutive periods overlap if $\WA<\WS$; the window is then called \textit{sliding} and a tuple can fall into several window \textit{instances}.
\item [\draft{Parameter Inputs (\I{})}] is the number of $O$'s input streams, one for each of $O$'s upstream peers.
\item [\draft{Single Key-by function $\bm{\keybysingleinmath{}}$}] \draft{extracts exactly one key from a tuple $t$, usually returning a subset of $t.\varphi$~\cite{akidau2015dataflow}.} $O$ maintains distinct window instances for each of its $\I$ input streams and for groups of tuples that share the same \textit{key}.
\item [\draft{Parameter Window Type ($\bm{\psipar}$)}] defines how window instances are internally maintained by $O$. If $\psipar = \text{single}$, a single instance is maintained per key and updated based both on tuples entering and leaving it. If $\psipar = \text{multi}$, multiple overlapping window instances are maintained per key and updated according to the incoming tuples. Hence, new window instances are continuously created while old ones are eventually discarded. As discussed in~\cite{geethakumari2019time}, $\psipar = \text{single}$ is preferable when $\WA\ll \WS$. \autoref{fig:windows_examples} shows how two operators (that only differ in $\psipar$) maintain their $w$ instances (each $w$ contains the tuples falling in it).
As shown, for each key $O$ maintains a list of sets, each with $I$ windows (one single set if $\psipar = \text{single}$, or one or more sets if $\psipar = \text{multi}$).
\item [\draft{Parameter Schema ($\bm{S}$)}] defines $O$'s output tuples.
\item [Functions $\bm{f_1,f_2,\ldots}$] are operator-specific functions to update $O$'s state and produce output tuples.
\end{description}
 
\begin{figure}[ht!]
\includegraphics[width=1\linewidth]{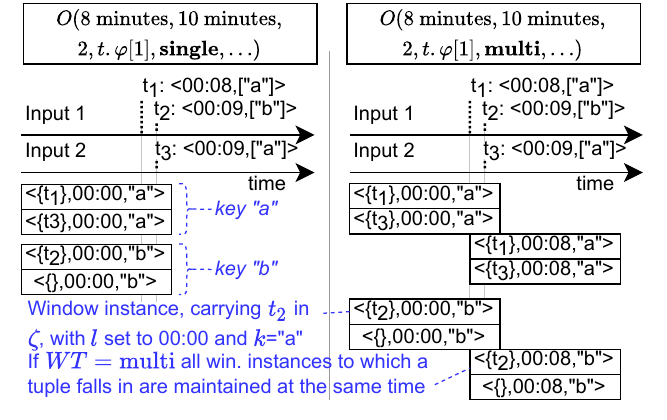}
\caption{Sample $w$ instances maintained by two $O$ operators with $\psipar = \text{single}$ (left) and $\psipar = \text{multi}$ (right).}\label{fig:windows_examples}
\end{figure}

In the following, we use $\langle \zeta,l,k \rangle$ for the combined notation of 
a window instance $w$, where $\zeta$ is $w$'s internal state (e.g., the tuples falling in $w$), $l$ is the event time of $w$'s left boundary (inclusive), and $k$ is $w$'s key. The right boundary of $w$ (exclusive) is computed as $w.l+\WS$.
As common in related works~\cite{flink,storm,beam}, when an output tuple $t_{out}$ is created in connection to a window instance $w$, $t_{out}.\tau$ is set to $w$'s right boundary. Since $w$'s right boundary is exclusive:
\begin{observation}
\label{obs:outtuplestimestamp}
Any output tuple $t_{out}$ produced from a window instance to which $t_{in}$ falls in is such that $t_{out}.\tau>t_{in}.\tau$.
\end{observation}

The Aggregate $A(\WA,\WS,1,\keybysingle,\psipar,S,f_A,f_R)$
defines $f_A$ to aggregate the tuples falling in one window instance $w$ into the $\varphi$ attribute of the output tuple created for $w$, and $f_R$ to incrementally update $w.\zeta$.
As we show in \autoref{sec:generalizedmodel}, \aggop{} can produce an output tuple both incrementally, updating its state for each new tuple it receives, or only upon the expiration of a window instance.
The Join $J(\WA,\WS,2,\keybysingle,\psipar,S,f_J)$ defines $f_J$ to match pairs of tuples (one from each stream) that fall to window instances sharing the same boundaries and associated with the same key.
Each matched pair can result in up to one output tuple.

\subsection{Shared-nothing parallelism (model and notation)}
\label{ssc:snparallelism}

\begin{figure}[ht!]
\includegraphics[width=\linewidth]{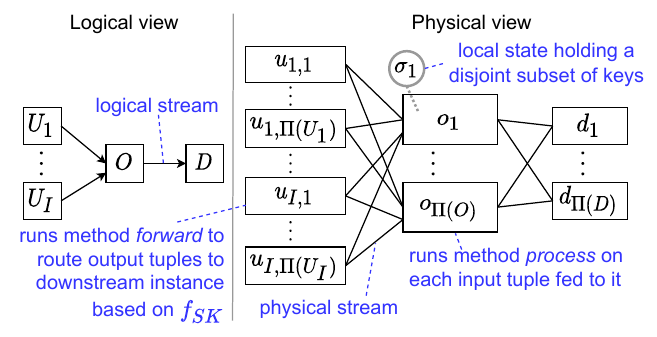}
\caption{Logical/physical views of $O$ with SN parallelism.}\label{fig:sn}
\end{figure}

SPEs let users define operators such as \aggop{}/\joinop{} as \textit{logical}, and later convert them into \textit{physical instances} (\autoref{fig:sn}).
We use $U_i$ and $D$ to refer to a generic upstream ingress/operator and downstream operator/egress, respectively.
$\Pi(X)$ denotes the parallelism degree of the operator, ingress, or egress $X$.
Each logical stream connecting a pair of logical operators (or an ingress/egress) is converted into one or more physical streams.
Each instance $o_j$ of operator $O$ is assigned a subset of the keys observed in their input streams (and thus windows), which it maintains in its local state $\opstate_j$. 

With SN parallelism, each instance $o_j$ has a dedicated queue for the tuples of each of its physical input or output streams.
Moreover, $o_j$ owns its dedicated $\opstate_j$, which is not concurrently accessed by other operator instances.
From an implementation perspective, all the operations carried out by $o_j$ are encapsulated in two methods: \texttt{process} and \texttt{forward}.
Method \texttt{process} encapsulates $o_j$'s analysis, and runs every time a tuple forwarded by an $u_{i,j}$ instance is available for processing, being $u_{i,j}$ the $j$-th instance of $U_i$.
Method \texttt{forward} encapsulates the routing of tuples to downstream instances and runs every time one or more tuples produced by \texttt{process} should be sent to a $D$ instance.
\draft{To route tuples, method \texttt{forward} has access to a \textit{mapping function} \opmappinginmath{} that maps output tuples' keys to $D$ instances according to $D$'s semantics.} We say $o_j$ is responsible for key $k$ if, given the $\opmapping{}$ used by $U_i$'s method \texttt{forward}, $\opmapping(k)=j$.

\subsection{Correctness conditions}
\label{ssc:correctness}

When deploying and running operators, users expect SPEs to enforce operators' semantics correctly.
For operator $O$, correct execution within an SPE can be defined as follows.
\begin{definition}
\label{def:correctness}
$O$'s instances execution is correct if, according to $O$'s specifications, any subset of $O$'s input tuples that could jointly contribute to an output tuple is indeed processed together and results in such an output tuple (if any). 
\end{definition}

For an Aggregate \aggop{}, Definition~\ref{def:correctness} implies that all the input tuples falling into one specific window instance $w$ should be jointly processed by $f_A$ and/or used by $f_{R}$ to update $w.\zeta$.
For a Join \joinop{}, it implies that any pair of tuples (\draft{one from the \textit{left} stream $L$ and one from the \textit{right} stream $R$}) $\left<t_L,t_R\right>$ such that $t_{L}\in w_{L}, t_{R} \in w_{R}, w_{L}.l=w_{R}.l, \text{ and } \keybysingle(t_{L})=\keybysingle(t_{R})$ is processed by $f_J$.

As discussed in~\cite{gulisano2020role} the correct execution for an instance $o_j$ requires to consistently maintain its \textit{watermark} \watermarkofinmath{o_j}:
\begin{definition}
\label{def:watermark}
The watermark \watermarkofinmath{o_j} of operator instance $o_j$ at a point in wall-clock time\footnote{From here on, we only differentiate wall-clock time (or simply time) from event time if such distinction is not clear from the context.} $\omega$ is the earliest event time a tuple $t^\ell$ to be processed by $o_j$ can have from time $\omega$ on (i.e., $t^{\ell}.\tau \ge \watermarkof{o_j}, \forall t^\ell \text{ processed from } \omega \text{ on}$).
\end{definition}

We say a window instance is \textit{expired} if its right boundary falls before the operator instance's watermark.
We say a tuple is expired if it only contributes to expired window instances.
Based on \watermarkofinmath{a_j}, an instance $a_i$ of \aggop{} can safely invoke $f_A$ for any $w$ such that $w.l+\WS\leq \watermarkof{a_j}$ since no more tuples that fall in $w$ will be received or processed by $f_R$, and it is thus safe to shift (if $\psipar = \text{single}$) or discard (if $\psipar = \text{multi}$) $w$ (cf. \autoref{sec:dsbasics}).
\draft{Appendix~\ref{apx:exampleofwatermark} builds on~\autoref{fig:windows_examples} to provide an example visualizing expired and non-expired windows.}
Based on \watermarkofinmath{j_i}, an instance $j_i$ of \joinop{} can safely
shift/delete $w$ if $w.l+\WS\leq \watermarkof{j_i}$, since $j_i$ will no longer invoke $f_J$ on $w$'s expired tuples.
In related work, \watermarkofinmath{o_j} is updated based on one of the following ways.

\subsubsection*{Implicit watermarks}
The first way to update watermarks assumes that each $o_j$'s physical input stream is timestamp-sorted~\cite{gulisano2012streamcloud}.
The tuples of such physical streams are merge-sorted in timestamp-order and fed to $o_j$ once \textit{ready}, as defined next based on~\cite{cederman2013concurrent}:
\begin{definition}
\label{def:ready}
\draft{$t^\ell_i$, the $\ell$-th tuple from timestamp-sorted stream $i$, is \emph{ready} to be processed if $t^\ell_i.\tau \leq \min_\text{i}\{\max_{m}(t^{m}_\text{i}.\tau)\}$,
the minimum among the latest $m$-th tuple timestamps from each timestamp-sorted stream $i$.}
\end{definition}
\noindent In such a case, \watermarkofinmath{o_j} can be safely updated by $o_j$ to $t.\tau$ for each incoming ready tuple $t$ fed to $o_j$.

\subsubsection*{Explicit watermarks}
The second way to update watermarks~\cite{flink,gulisano2020role} assumes ingresses/operators periodically propagate watermarks through the DAG as special tuples.
This allows handling both out-of-timestamp-order streams as well as timestamp-sorted streams whose rate might drop to zero during periods of time (in the latter case this could affect the sorting performed when relying on implicit watermarks).
Upon receiving a watermark, $o_j$ stores the watermark's time, updates \watermarkofinmath{o_j} to the minimum of the latest watermarks received from each input stream, and propagates \watermarkofinmath{o_j}.

Note that, while both options support correct execution, implicit watermarks also ensure streams are fed in total order to $o_j$. Hence, $o_j$ can seamlessly support timestamp-order-sensitive analysis.
Explicit watermarks are only safe for timestamp-order-insensitive analysis or require $o_j$ to sort its input tuples before processing them~\cite{gulisano2020role,akidauwatermarks}.

\subsection{The ScaleGate object}
\label{ssc:scalegate}

\draft{\emph{ScaleGate}~\cite{scalejoin} (SG) is a shared data object, that in this work is extended and used to support \xxx{}'s algorithmic implementation.}
It allows several \textit{sources} to concurrently and efficiently merge timestamp-sorted streams into a timestamp-sorted stream of ready tuples (cf.~Definition \autoref{def:ready}).
Also, it allows several \textit{readers} to consume all the ready tuples of the latter stream. Its lock-free, linearizable algorithmic implementation supports efficient deterministic processing~\cite{slackscalegate,walulya2018viper}.
\draft{A fixed set of sources/readers can interact with an SG object using the API methods:}
\begin{itemize}[leftmargin=*]
\item \texttt{addTuple(tuple,i)}: which merges a \texttt{tuple} from the $i$-th source in the timestamp-sorted stream of ready tuples.
\item \texttt{getNextReadyTuple(i)}: which provides to the $i$-th reader the next earliest ready tuple that has not been yet consumed by it. Note that each tuple, once ready, will be returned to all readers invoking the method.
\end{itemize}
To ease notation, we refer to methods \texttt{addTuple} and \texttt{getNextReadyTuple} as \texttt{add} and \texttt{get} in the remainder.

\subsection{Elasticity}
\label{ssc:elasticity}
The computational cost of a stream processing application varies over time~\cite{gulisano2012streamcloud}.
Hence, an execution in which the parallelism degree of operators is fixed can lead to imbalances in the work of such operators.
When the overall work of an operator is unbalanced but can still be carried out by its instances, a \textit{load balancing} reconfiguration is needed to change the work distribution (i.e., the mapping of each key and the instance responsible for it).
If new instances are to be \textit{provisioned}, the new work distribution re-assigns some keys to the newly allocated instances.
Since over-provisioned systems can lead to high latency~\cite{walulya2018viper} and unnecessary costs~\cite{gulisano2012streamcloud}, existing instances should also be \textit{decommissioned} when fewer instances suffice for the overall workload.
We use the term \textit{reconfiguration} to refer to any of these actions
(note provisioning and decommissioning imply load balancing).\looseness=-1

\separate
\section{Problem Definition and Approach}\label{sec:smps}

This work focuses on intra-process parallel and elastic execution of stateful stream processing analysis for time-based sliding windows (or simply windows), thus assuming $\WA<\WS$.
For any given stateful operator $O$, we do not make any assumption on the frequency, periodicity, or timings with which tuples are fed to it. As such, we cannot infer the event time of the $\ell+1^{th}$ input tuple fed to an instance of $O$ based on that of the $\ell^{th}$ tuple.
We assume the information needed by an instance $o_j$ to update its watermark \watermarkofinmath{o_j} is carried by tuples' metadata (cf.~\autoref{sec:dsbasics}), be it through $t.\tau$ and implicit watermarks, or by forwarding explicit watermarks (as special tuples or additional metadata of regular tuples).
We assume ingresses/operator instances are continuously delivering tuples/watermarks and have completed their bootstrap phase (i.e., all have started delivering tuples/watermarks).
We assume the threads in charge of the physical execution of $O$ share memory (physical/logical) and can read/write shared objects stored in it.

\draft{Within this setup, we first show that SN parallelism might require data duplication to parallelize the execution of operators whose semantics are more general than those of \aggop{}/\joinop{} (e.g., those of our running example from~\autoref{sec:introduction}).}
We then prove that for such generalized semantics, implementations that rely on SN parallelism (with an arbitrary degree of data duplication) can be transformed to semantically equivalent implementations that rely on shared memory and do not incur the overheads of data duplication nor those of state transfer during elastic reconfigurations.
Finally, we also provide an extensive discussion and evaluation (with several state-of-the-art baselines) based on a fully implemented prototype.
For ease of presentation, we center most of our discussions around a single stateful operator.
Nonetheless, 
\xxx{} can support the execution of multiple stateful operators within one query,  
as we discuss in~\autoref{sec:stretch} and~\autoref{sec:implementation}.

\draft{When it comes to elasticity, note that \xxx{} does not aim at embedding a specific policy about when/how to balance load or provision/decommission instances (e.g., based on energy~\cite{de2017proactive} or CPU consumption~\cite{gulisano2012streamcloud}), but rather defines a generic API for external modules. We show in \autoref{sec:eval} the use of \xxx{} in conjunction with two such modules.}

\separate
\section{A Generalized Stateful Operator}
\label{sec:generalizedmodel}

As introduced in~\autoref{sec:introduction}, data duplication is a potential drawback of SN parallelism.
We prove herein that there exist operators whose semantics are richer than those of \aggop{} and \joinop{} (cf.~\autoref{sec:prel}), and which might need data duplication for some of their tuples.
\draft{To frame the type of stateful operator modeled by \xxx{}, we then introduce a unified and \textit{generalized operator} \generalopplusinmath{}, later used when arguing about the semantic equivalence between SN and VSN setups.}
With \generalopplusinmath{} each input tuple $t$ can be shared with an arbitrary number of instances, thus accounting for any data duplication level.

\subsection{SN parallelism and data duplication}

Before proving the data duplication need, we introduce the following lemma and definitions.

\begin{lemma}
\label{lem:consecutive}
Let $\mathbb{W}$ be the set of non-expired window instances held by $o_j$ for key $k$.
The next tuple with key $k$ fed to $o_j$ can fall in a window instance $w$ such that $w \in \mathbb{W}$.
\end{lemma}
\begin{proof}
Being $w^*$ one of the latest non-expired window instance in $\mathbb{W}$ (i.e., one with the highest left boundary $l$) and $t^*$ the next tuple fed to $o_j$, 
if the statement does not hold, then either (i) $\watermarkof{o_j} \geq w^*.l+\WS$, which contradicts that $\mathbb{W}$ contains non-expired window instances, or (ii) $t^*.\tau \geq w^*.l+\WS$, which contradicts our assumption from~\autoref{sec:smps} about future tuples' timestamps being not known.
\end{proof}

\begin{definition}
\label{def:fpluskb}
\draft{$\keybymulti(t)$ is a Multi Key-by function that, differently from $\keybysingle(t)$, does not result in one key but rather in a set (possibly empty) of keys.}
\end{definition}

\begin{definition}
\label{def:aplusjplus}
$\aggplusop$ and $\joinplusop$ are generalizations of \aggop{} and \joinop{}, respectively, for which \keybymultiinmath{} replaces \keybysingleinmath{}.
\end{definition}

Given Lemma~\ref{lem:consecutive} and Definitions~\ref{def:fpluskb} and \ref{def:aplusjplus}, we can introduce our theorem about data duplication.

\begin{theorem}
\label{thm:notgeneral}
If any arbitrary pair of tuples $t^\ell,t^m$ fed to \aggplusopinmath{}/\joinplusopinmath{} can share a key, then \aggplusopinmath{}/\joinplusopinmath{} might need data duplication for SN parallelism to correctly enforce \aggplusopinmath{}/\joinplusopinmath{}'s semantics.
\end{theorem}

\begin{proof}
Let us begin with \aggplusopinmath{} and let us assume $t^\ell$ is the last tuple forwarded to $a_+$, one of \aggplusopinmath{}'s parallel instances, and that $t^{\ell+1}$ is the next tuple that has to be forwarded to one of \aggplusopinmath{}'s instances.
If $t^\ell$ and $t^{\ell+1}$ indeed share one key,
and given Definition~\ref{def:correctness} and Lemma~\ref{lem:consecutive}, then $t^{\ell+1}$ must be sent to $a_+$ to correctly enforce \aggplusopinmath{}'s semantics.
The same holds for $t^{\ell+2},t^{\ell+3},\text{etc.}$
To avoid any data duplication, all data should be sent to $a_+$ to account for any future tuple that could share a key with a previous tuple. In this case, though, \aggplusopinmath{} would not run in parallel.
The same reasoning holds for \joinplusopinmath{} and for two tuples $t_{R}, t_{L}$ (one from stream $R$, one from $L$) that should be compared when sharing a common key.
\end{proof}

To support the following discussion about our generalized stateful operator \generalopplusinmath{}, we give the following corollary. 
\begin{corollary}
\label{cor:duplicationviamap}
\aggplusopinmath{} and \joinplusopinmath{} semantics can be enforced combining \mapop{}, \aggop{} and \joinop{} operators with SN parallelism.
\end{corollary}

This is so since by preceding \aggop{} and \joinop{} with \mapop{} operators, \mapop{} can be used to process each incoming tuple $t^\ell$ so that (1) as many copies as keys in $t^\ell$ are made of $t^\ell$, (2) each such copy $t'^\ell$ carries one of $t^\ell$'s keys in $t'^\ell.\varphi$ as an extra attribute, and (3) such attribute is then returned by the \keybysingleinmath{} of \aggop{} or \joinop{}.

To further build on our running example (cf. \autoref{sec:introduction}), note that the operator computing the longest tweet per-hour and per-hashtag is an \aggplusopinmath{} operator.
In this case, \keybymultiinmath{} would extract each tweet's hashtag as a key, and $f_{A}$ would find the longest tweet of each such key.
\draft{Corollary~\ref{cor:duplicationviamap} shows how the programmer intervention is required to enforce \aggplusopinmath{}'s semantics, in this case expressing \aggplusopinmath{} as an \mapop{}/\aggop{} combination (we refer the reader to Appendix~\ref{apx:appendixcorollary} for an extended example).}

\subsection{\xxx{}'s generalized model for stateful analysis}
\label{ssc:stretchgeneralizedmodel}

After discussing \aggplusopinmath{}/\joinplusopinmath{} generalizations in Theorem~\ref{thm:notgeneral}, we now introduce a general operator \generalopplusinmath{}, and a model for SN parallelism that (1) 
encapsulates the semantics of \aggplusopinmath{}/\joinplusopinmath{} (and thus that of \aggop{} and \joinop{}, cf. Corollary~\ref{cor:duplicationviamap}), and 
(2) allows for an arbitrary degree of data duplication.

The logical and physical representations of \generalopplusinmath{} resemble those shown in~\autoref{fig:sn}.
\generalopplusinmath{} has $I$ upstream peers $U_1,\ldots,U_I$ (which could be an arbitrary mixture of operators and ingresses) and one downstream peer $D$ (an operator or an egress).
We refer to the $j$-th instance of $U_i$, \generalopplusinmath{}, and $D$ as
$u_{i,j}$,  $\generalopplusinst_j$, and $d_j$, respectively; $q_{a,b}$ is the queue between instances $a$ and $b$. 
\generalopplusinmath{} is defined as follows:
$$
\generalopplus(\WA,\WS,I,\keybymulti,\psipar,S,\opmapping,f_U,f_O,f_S)
$$

\draft{
Besides the aforementioned parameters $\WA$, $\WS$, $I$, $\keybymulti$, $\psipar$, $S$, and $\opmapping$, \generalopplusinmath{} defines an \textit{update} ($f_U$), an \textit{output} ($f_O$), and a \textit{slide} function ($f_S$) to maintain its window instances:
$f_U$ is invoked, upon the reception of $t$, to update the state of the instances of keys associated with $t$ and (optionally) to return a set of tuple payloads to be forwarded to $D$; $f_O$ is invoked when a set of instances expires, to return a set of tuple payloads
to be forwarded to $D$; $f_S$ is invoked when instances slide, to return the updated states for a set of $I$ window instances that have just advanced by $\WA$ time units.
}

\begin{table}[t]
\caption{Functions \generalopplusinmath{} uses to maintain window instances.}
\label{tab:opo_functions}
\begin{tabular}{@{}llp{1.7cm}p{3.2cm}@{}}
\toprule
\textbf{} & \textbf{Input}         & \textbf{Output}    &   \textbf{Default behavior}  \\
\midrule
$f_U$             & $\{w_1,\ldots,w_I\},t$ & $\{\zeta_1,\ldots,\zeta_I,$ $\varphi^1,\ldots,\varphi^\ell\}$ & Store $t$ in $w.\zeta$ of $t$'s sender. Returns no $\varphi$.\\
$f_O$             & $\{w_1,\ldots,w_I\}$   & $\{\varphi^1,\ldots,\varphi^\ell\}$ & Returns no $\varphi$.\\
$f_S$             & $\{w_1,\ldots,w_I\}$ & $\{\zeta_1,\ldots,\zeta_I\}$ & Purge stale tuples. \\
\bottomrule
\end{tabular}
\end{table}

\begin{algorithm}[b]
\footnotesize
\SetAlgoLined
\DontPrintSemicolon
\SetKwInput{KwParameters}{Instance-local variables}
\SetKwInput{KwAuxiliaryFunctions}{Auxiliary methods}
\SetKwProg{Proc}{Method}{}{}

\KwParameters{}
$\keybymulti,\opmapping$ \tcp{From \generalopplusinmath{}'s  definition}

\BlankLine
\Proc{\FuncSty{forwardSN}($\{t^1,\ldots,t^\ell\}$)}
{
    \For{$t \in \{t^1,\ldots,t^\ell\}$}{
        $\mathbb{P} \xleftarrow{} \{\}$ \tcp{empty set of downstream peers}
        $\mathbb{K} \xleftarrow{} \keybymulti (t)$ \tcp{get keys of $t$}\label{alg:u:getkeys}
        \lFor{$k \in \mathbb{K}$}{
            $\mathbb{P} \xleftarrow{} \mathbb{P} \cup \opmapping(k)$\label{alg:u:findinst}
        }
        \lFor{$p \in \mathbb{P}$}{
            $q_{u_{i,j},\generalopplusinst_p}$.\FuncSty{add($t$)}\label{alg:u:sendt}
        }
    }
}
\caption{Method \texttt{forwardSN} (\texttt{forward} for SN setups) for a $u_{i,j}$ instance. Invoked when a non-empty set of tuples is to be forwarded to \generalopplusinmath{}.}
\label{alg:u}
\end{algorithm}

\begin{algorithm*}[h]
\footnotesize
\SetAlgoLined
\DontPrintSemicolon

\SetKwInput{KwParameters}{Instance-local variables (besides $\WA,\WS,I,\keybymulti,\psipar,S,\opmapping,f_U,f_O,f_S$)}
\SetKwInput{KwAuxiliaryFunctions}{Auxiliary methods}
\SetKwProg{Proc}{Method}{}{}

\KwParameters{}
$\sigma_j$ \tcp{$\generalopplusinst_j$'s state, initially empty}
$W$ \tcp{$\generalopplusinst_j$'s watermark, initially 0}
$\rho$ \tcp{earliest $w.l$ of any $w \in \sigma_j$, initially 0}

\BlankLine
\KwAuxiliaryFunctions{}
\FuncSty{update$W$($t$)} \tcp{update $W$ based on $t$'s metadata}
\FuncSty{$\sigma_j$.remove($k,\ell$)} \tcp{remove the $\ell$-th set of $I$ window instances for key $k$}
\FuncSty{$\sigma_j$.set($k,\ell,\{\zeta_1,\ldots,\zeta_I\}$)} \tcp{set states of the $\ell$-th set of $I$ window instances for key $k$}
\FuncSty{$\sigma_j$.shift($k,\ell,\{\zeta_1,\ldots,\zeta_I\}$)} \tcp{for the $\ell$-th set of $I$ $w$ inst. for key $k$, increase $w.l$ by $\WA$ and set states}
\FuncSty{$\sigma_j$.check\&Create($k,l$)} \tcp{add set of $I$ $w$ instances for key $k$ and $w.l=l$ if they are not already in $\sigma_j$}
\FuncSty{earliestWinL($t$)} \tcp{get earliest $w.l$ for any $w$ in which $t$ falls}
\FuncSty{latestWinL($t$)} \tcp{get latest $w.l$ for any $w$ in which $t$ falls}
\FuncSty{prepareOutTuples($\{\varphi^1,\ldots,\varphi^\ell\}$)} \tcp{create an output tuple and set its metadata for each of the $\ell$ payloads}

\BlankLine
\Proc(\tcp*[h]{forward results (if any) and shift/remove the $w$ instances}){\FuncSty{forwardAndShift}($k$)} {
    \FuncSty{forward}(\FuncSty{prepareOutTuples}($f_O$($\sigma_j[k][1]$)))\; \label{alg:o:output}
    \If{$\psipar = \text{single}$}{\label{alg:o:shiftstart}
        $\{\zeta_1,\ldots,\zeta_I\} \xleftarrow{} f_S(\sigma_j[k][1])$\;
        \lIf{$\exists i \in \{1,\ldots,I\}|\zeta_i\neq \emptyset$}{
            \FuncSty{$\sigma_j$.shift($k,1,\{\zeta_1,\ldots,\zeta_I\}$)}\label{alg:o:shiftend}
        }
        \lElse{
            \FuncSty{$\sigma_j$.remove($k,1$)}\label{alg:o:remove1}
        }
    }
    \lElse{
        \FuncSty{$\sigma_j$.remove($k,1$)}\label{alg:o:remove2}
    }
}

\BlankLine
\Proc{\FuncSty{handleInputTuple}($t$)} {
    \If{$\exists k | k \in \keybymulti(t) \wedge \opmappingfun(k)=j$}{
        $\tau_1 \xleftarrow{} $\FuncSty{earliestWinL($t$)} \tcp{Find $w.l$ of window instance to update} \label{alg:o:findwinstart}
        \lIf{$\psipar = \text{single}$}{
            $\tau_2 \xleftarrow{} $\FuncSty{earliestWinL($t$)}
        }
        \lElse{
            $\tau_2 \xleftarrow{} $\FuncSty{latestWinL($t$)}
        }
        \lIf{$\tau_1<\rho$}{
            $\rho \xleftarrow{} \tau_1$\label{alg:o:checkrho}
        }
        \For(\tcp*[h]{Create/update window instance}){$\ell \in \{0,\ldots,(\tau_2-\tau_1)/\WA\}$}{\label{alg:o:forwin}
            \For{$k \in \keybymulti(t) | \opmappingfun(k)=j$}{\label{alg:o:forkey}
                \FuncSty{$\sigma_j$.check\&Create($k,\tau_1+\ell*\WA$)}\;\label{alg:o:updatestart}
                $\{\zeta_1,\ldots,\zeta_I,\varphi^1,\ldots,\varphi^\ell\}\xleftarrow{}f_U(\sigma_j[k][\ell])$\;\label{alg:o:f_u}
                \FuncSty{forward}(\FuncSty{prepareOutTuples}($\varphi^1,\ldots,\varphi^\ell$)) \tcp{forward results (if any)} 
                \FuncSty{$\sigma_j$.set($k,\ell,\{\zeta_1,\ldots,\zeta_I\}$)}\;\label{alg:o:updateend}
            }
        }
    }
}

\BlankLine
\Proc{\FuncSty{processSN}($t$)}
{\label{alg:o:process}

    \FuncSty{update$W$($t$)} \tcp{update $\generalopplusinst_j$ watermark}\label{alg:o:water}
    
    \While(\tcp*[h]{while $\generalopplusinst_j$ has expired window instances}){$\rho+\WS<W$}{\label{alg:o:outstart}
        \lWhile(\tcp*[h]{and an expired window inst. $w$ starts at $\rho$, handle $w$}){$\exists k | \sigma_j[k][1].l = \rho$}{\label{alg:o:checkearlieast}
            \FuncSty{forwardAndShift}($k$)
        }
        $\rho\xleftarrow{}\rho+\WA$ \tcp{update $\rho$ value}\label{alg:o:outend}
    }
    \FuncSty{handleInputTuple}($t$)\;
}
\caption{Method \texttt{processSN} (\texttt{process} for SN setups) for an $\generalopplusinst_j$ instance. Runs when $t$ is returned by $q_{u_{i,j},\generalopplusinst_j}$.}
\label{alg:o}
\end{algorithm*}

As shown in \autoref{tab:opo_functions}, a default behavior is associated with each function. To draw a parallel with Object-Oriented Programming, \generalopplusinmath{} acts as a \textit{superclass} of operators like \aggplusopinmath{} and \joinplusopinmath{} (we formally prove this later in the section). The latter can be instantiated by specializing some of its functions $f_U$, $f_O$, and $f_S$.
The user is expected to specialize at least one function, since the default implementations result in \generalopplusinmath{} maintaining all tuples that fall in each window instance (via $f_U$ and $f_S$) but without producing any output tuple.

\subsubsection*{Implementation of \texttt{forward} and \texttt{process} methods}

We start with the method \texttt{forwardSN} (\texttt{forward} for SN setups) run by $u_{i,j}$ instances (\autoref{alg:u}).
As soon as a tuple $t$ is ready to be forwarded to \generalopplusinmath{}, \texttt{forwardSN} retrieves the set of keys associated with $t$ (L\ref{alg:u:getkeys}).
For each key $k$, it then adds $t$ to the queue of each downstream peer responsible for at least one of $t$'s keys (L\ref{alg:u:findinst}-\ref{alg:u:sendt}).

We move now to the description of method \texttt{processSN} (\texttt{process} for SN setups, \autoref{alg:o}). 
Upon reception of $t$ (L\ref{alg:o:process}), the method updates $\generalopplusinst_j$'s watermark (L\ref{alg:o:water}) with the information contained in  $t$'s metadata (if any, cf.~\autoref{sec:prel}). It then proceeds by outputting the result of all expired window instances and shifting/removing them (L\ref{alg:o:outstart}-\ref{alg:o:outend}).
We use the notation $\sigma_j[k][\ell]$ to refer to the $\ell$-th set of $I$ window instances (one for each of the $I$ upstream operators) maintained in $\sigma_j$ for key $k$.
It starts with the earliest ones, having $\rho$ as left boundary (L\ref{alg:o:checkearlieast}), outputs their result invoking $f_O$ (L\ref{alg:o:output}) and shifts them (when $\psipar = \text{single}$) invoking $f_S$ (L\ref{alg:o:shiftstart}-\ref{alg:o:shiftend}) and/or removes them (L\ref{alg:o:remove1}-\ref{alg:o:remove2}), if $\psipar = \text{multi}$ or if all window instances have an empty state. When no window instances with left boundary $\rho$ remain, the method continues checking the ones starting at $\rho+\WA$ (L\ref{alg:o:outend}).
After taking care of expired window instances, the method proceeds identifying the set of window instances to which $t$ falls in (depending on $\psipar$) and adjusting $\rho$ if needed (L\ref{alg:o:findwinstart}-\ref{alg:o:checkrho}). For each such window instance (L\ref{alg:o:forwin}), and for each of $t$'s keys responsibility of \generalopplusinstinmath{} (L\ref{alg:o:forkey}) the method creates the corresponding window instances, if not already defined, and updates their state invoking $f_U$ (L\ref{alg:o:updatestart}-\ref{alg:o:updateend}).

Upon invoking $f_U$/$f_O$, \generalopplusinmath{} sets the event time of any resulting output tuple to the right boundary of the window instances on which $f_U$/$f_O$ have been invoked, as done by \aggop{} and \joinop{} (see~\autoref{sec:prel}), using method \texttt{prepareOutTuples}.

\subsubsection*{Formal Guarantees}

After covering the implementation of \generalopplusinmath{}, we now prove \generalopplusinmath{} encapsulates the semantics of \aggplusopinmath{}/\joinplusopinmath{} (and thus \aggop{}/\joinop{}).

\begin{table*}[ht!]
\centering
\caption{API of the $\Xplus$ shared data structure. Highlighted methods are only relevant for elastic setups.}
\label{tab:xapi}
\begin{tabular}{@{}lp{15.0cm}@{}}
\toprule
API Method & Description \\
\midrule
\textbf{\texttt{add($\bm{t,j}$)}} & invoked by source $j$ to add tuple $t$. \\
\textbf{\texttt{get($\bm{j}$)}} & invoked by reader $j$ to retrieve the next tuple conforming to the watermarks that are also delivered by the \texttt{get} method. \\
\rowcolor{Gray}
\textbf{\texttt{addReaders($\bm{\mathbb{R},j}$)}} & invoked by reader $j$ to add readers in $\mathbb{R}$ that are not already readers of $\Xplus$. Once invoked, $\Xplus$ delivers to $\mathbb{R}$ tuples and watermarks starting from the ones that will be also returned to $j$ once $j$ invokes \texttt{get}. If more readers invoke \texttt{addReaders} concurrently, only one succeeds. Returns \texttt{true} only if it adds all new readers in $\mathbb{R}$. \\
\rowcolor{Gray}
 \textbf{\texttt{removeReaders($\bm{\mathbb{R}}$)}} & removes each of the readers in $\mathbb{R}$ (that are readers of $\Xplus$) when invoked by a reader of $\Xplus$. If more readers invoke \texttt{removeReaders} concurrently, only one succeeds. Returns \texttt{true} only if it removes all existing readers in $\mathbb{R}$. \\
\rowcolor{Gray}
 \textbf{\texttt{addSources($\bm{\mathbb{S}}$)}} & adds each of the sources in  $\mathbb{S}$ that are not already sources of $\Xplus$ when invoked by a source of $\Xplus$. If more sources invoke \texttt{addSource} concurrently, only one succeeds. Returns \texttt{true} only if it adds all new sources in $\mathbb{S}$. \\ 
\rowcolor{Gray}
 \textbf{\texttt{removeSources($\bm{\mathbb{S}}$)}} & removes each of the sources in $\mathbb{S}$ (that are sources of $\Xplus$) when invoked by a source of $\Xplus$. If more sources invoke \texttt{removeSources} concurrently, only one succeeds. Returns \texttt{true} only if it removes all existing sources in $\mathbb{S}$. \\
\bottomrule
\end{tabular}
\end{table*}

\begin{theorem}
\label{thm:generalizationworks}
\generalopplusinmath{} guarantees \aggplusopinmath{} and \joinplusopinmath{} semantics (cf. Def.~\ref{def:aplusjplus}).
\end{theorem}

\begin{proof}[Proof]
The method \texttt{forwardSN} of $u_{i,j}$ implements the same semantics of those of an additional map, placed after each $u_{i,j}$ instance, that create copies of each tuple according to Corollary~\ref{cor:duplicationviamap}. Moreover, with \generalopplusinmath{}'s, \aggop{} can be instantiated by setting $I=1$ and using $f_A$ as $f_O$ and/or $f_R$ as $f_S$. Similarly, \joinop{} can be instantiated by setting $I=2$ and matching tuples via $f_J$ either incrementally with $f_U$ or upon window instance expiration with $f_O$.
\end{proof}

An additional property of \generalopplusinmath{} is about how watermarks can be delivered by its instances, as discussed next.

\begin{lemma}\label{lem:deliveringoutputs}
For each \generalopplusinstinmath{} instance, output tuples timestamps' represent valid implicit watermarks, since $t^m.\tau \geq t^\ell.\tau$  for any pair of consecutive output tuples $t^\ell$ and $t^m$.
\end{lemma}

\begin{proof}[Proof]
As shown in~\autoref{alg:o} L\ref{alg:o:outstart}-\ref{alg:o:outend}, \generalopplusinstinmath{} produces output tuples in timestamp order. As such, each output tuple's timestamp is also a valid watermark, since no tuple $t^{\ell+1}$ delivered after $t^\ell$ can have $t^{\ell+1}.\tau<t^{\ell}.\tau$.
\end{proof}

\draft{
We refer to Appendix~\ref{apx:examples} for several complete operator examples, including one for an
$\aggplusop$ running the example from~\autoref{sec:introduction} and one for ScaleJoin~\cite{scalejoin} (a $\joinplusop$ that we later use in \autoref{sec:eval}), and to Appendix~\ref{apx:appendixexecution} for an additional example that connects to Theorem~\ref{thm:generalizationworks} and shows how an execution of an $\aggplusop$ results in the same state updates observed when $\aggplusop$'s semantics are implemented using $\mapop{}$ and $\aggop{}$ operators.
}

Despite not being within the scope of this work to create a complete taxonomy of all the stateful analysis semantics \generalopplusinmath{} can cover, but rather to show \generalopplusinmath{} is a generalization of common stateful operators such as \aggop{} and \joinop{}, note that, in the spirit of a description of \generalopplusinmath{} semantics:
\begin{enumerate}
    \item differently from \aggop{} and \joinop{}, \generalopplusinmath{} can work with an arbitrary number of upstream peers each producing tuples with their own schema, and
    \item \generalopplusinmath{} can implement custom stateful operators (e.g., the ScaleJoin operator, presented in Appendix~\ref{apx:examples} and~\autoref{sec:eval}).
\end{enumerate}

\separate
\section{VSN Parallelism and Elasticity}
\label{sec:stretch}

After introducing \generalopplusinmath{}, we show \generalopplusinmath{} can leverage shared memory to avoid data duplication and state transfers during elastic reconfigurations while preserving its semantics. 
We first focus on how VSN parallelism can overcome the data duplication overhead.
For ease of exposition, we begin with a static setup in which $\Pi(\generalopplus)$ is fixed.

In \xxx{}, we assume that each pair of instances $\langle u_{i,j},\generalopplusinst_j \rangle$ has no dedicated queue (cf.~\autoref{ssc:snparallelism}), but rather that all $U_i$ and \generalopplusinmath{} instances share a single \textit{Tuple Buffer} (\Xinmath{}) object (cf.~\autoref{fig:vsn_elastic}) which behave according to the next definition.

\begin{definition}\label{def:xin}
\Xinmath{} is a data structure that allows a set of sources to concurrently add tuples to it, that delivers each tuple exactly once to each one of its readers, and that merges sources' watermarks into a single stream of non-decreasing watermarks, each delivered to all readers. It defines the methods presented in \autoref{tab:xapi}.
\end{definition}

\draft{We rely on a generic $\X$ data structure for generality. We discuss a specific data structure with such an API in \autoref{sec:esg_impl}.}

\subsubsection*{Implementation of the \texttt{forward} method}

In this new setup, $u_{i,j}$ instances run the method \texttt{forwardVSN} (\texttt{forward} method for VSN setups) in~\autoref{alg:u_vsn}. 
Each $u_{i,j}$ instance carries an $id$ (L\ref{alg:idparam}) that represents its index as source of $\X$ (i.e., 1 for $u_{1,1}$ and $\sum_i \Pi(U_i)$ for $u_{I,\Pi(U_I)}$) and passes such $id$ when adding tuples to $\X$ via the \texttt{add} method (L~\ref{alg:u_vsn_forstart}).
In this case, since all tuples are visible to each \generalopplusinstinmath{} instance, \generalopplusinstinmath{} should process only the tuples that carry at least one key that is its responsibility. 
Noting that, nonetheless, this is already the case in method \texttt{handleInputTuple} (\autoref{alg:o} L\ref{alg:o:forkey}), we make the following observation. 

\setlength{\algomargin}{\dimexpr\leftskip+\parindent}
\begin{algorithm}[h]
\footnotesize
\SetAlgoLined
\DontPrintSemicolon
\SetKwInput{KwParameters}{Instance-local variables}
\SetKwInput{KwAuxiliaryFunctions}{Auxiliary functions}
\KwParameters{}
$id$ \tcp{instance's id}\label{alg:idparam}

\BlankLine
\SetKwProg{Proc}{Function}{}{}
\Proc{\FuncSty{forwardVSN}($\{t_1,\ldots,t_m\}$)}
{
    \lFor{$t \in \{t_1,\ldots,t_m\}$}{\label{alg:u_vsn_forstart}
        $\X_{in}$.\FuncSty{add($t,id$)}\label{alg:u_vsn_forend}
    }
}
\caption{Method \texttt{forwardVSN} (\texttt{forward} for VSN setups) for a $u_{i,j}$ instance. Invoked when a non-empty set of tuples is to be forwarded to \generalopplusinmath{}.}
\label{alg:u_vsn}
\end{algorithm}

\begin{observation}
\label{thm:vsnstatic}
When $U_i$ and \generalopplusinmath{} instances are connected through $\X$ objects, and $U_i$ instances use~\autoref{alg:u_vsn} to forward tuples, \generalopplusinmath{} can run in parallel enforcing correctly \generalopplusinmath{}'s semantics without data duplication.
\end{observation}

\draft{Observation~\ref{thm:vsnstatic} focuses on the data duplication overhead for a single operator \generalopplusinmath{}. When considering a series of operators, note that tuples' data could be duplicated also when $D$ is stateful and \generalopplusinstinmath{} and $d_j$ instances share tuples through dedicated queues. 
Hence, for generality, we  also replace the queues of $\langle \generalopplusinst_j,d_j \rangle$ pairs with a $\X$, as shown in~\autoref{fig:vsn_elastic} (we name the objects $\X_{in}$ and $\X_{out}$ to differentiate them).
Note our model assumes one $D$ operator for simplicity; our results hold also if more $D$ operators invoke \texttt{get} on $\X_{out}$.}

\subsubsection*{From static to elastic setups}

We now show that if $\generalopplusinst_j$ instances share a global state $\sigma$ -- rather than per-instance $\opstate_j$ states -- this can enable state-transfer-free elasticity (cf. \autoref{sec:introduction}). We thus move from a static setup in which $\Pi(\generalopplus)$ and \opmappinginmath{} are fixed to one in which both can change over time.
It should be noted that, as we formally prove later in the section, although $\opstate$ can potentially be exposed to concurrent updates for distinct keys from the various $\generalopplusinst_j$ instances, \xxx{} ensures that each key is consistently updated by exactly one instance at a time.

A reconfiguration implies a change in \opmappinginmath{} to hold from a certain event time onward.
We use the term \textit{epoch} to refer to the event time period in between two watermarks during which the mapping of keys to instances does not change.
Hence, being $e$ the current epoch, $\mathbb{O}$ the set of \generalopplusinmath{} instances, and $\opmapping$ the mapping in $e$, a reconfiguration implies the start of a new epoch $e^{*}$ for which a new mapping $\opmapping^{*}$ is used for a (possibly different) set of operator instances $\mathbb{O}^{*}$.
We focus on a single transition from epoch $e$ to $e^{*}$ since subsequent epoch switches happen with the same logic.
To ease presentation, we define two temporary conditions:
\begin{itemize}[leftmargin=*]
    \item \textbf{Cond.~1:} For provisioning reconfigurations, the joining instances $\mathbb{O}^* \setminus \mathbb{O}$ are already instantiated, and start retrieving/processing tuples as soon as they are given access to $\Xplus_{in}$. For decommissioning reconfigurations, the leaving instances $\mathbb{O} \setminus \mathbb{O}^*$ will stop processing/outputting tuples once they are disconnected from $\Xplus_{in}$ and $\Xplus_{out}$.
    \item \textbf{Cond.~2:} All instances $\generalopplusinst_j \in \mathbb{O} \cup \mathbb{O}^*$ have access to a set of variables $\{e,e^*,\mathbb{O},\mathbb{O}^*,\opmappingfun^*,\gamma\}$ that represent the current epoch id number ($e$), the next epoch id number ($e^*$), the set of instances of the current epoch ($\mathbb{O}$), the set of instances of the next epoch ($\mathbb{O}^*$), the mapping function of the next epoch ($\opmappingfun^*$), and a shared event time ($\gamma$) greater than the current watermark of any instance $\generalopplusinst_j$, used by all $\generalopplusinst_j$ instances to trigger a reconfiguration as soon as their watermark is greater than or equal to $\gamma$. $\generalopplus$ instances receive special \textit{control} tuples, and set such parameters using a method named \texttt{prepareReconfig}. Control tuples can be distinguished with a method named \texttt{isControl} and are not processed to update \generalopplusinmath{}'s state.
\end{itemize}
\autoref{fig:vsn_elastic} shows \xxx{}'s setup.
We cover the actual coordination of instances, the implementation of \texttt{isControl} and \texttt{prepareReconfig}, and satisfy Cond. 1/2 in~\autoref{sec:implementation}.

\subsubsection*{Elasticity and correctness}

Before focusing on \xxx{}'s shared state and elasticity protocol, we want to draw attention to one challenging connection between elasticity and correctness guarantees.
A change in $\Pi(\generalopplus)$ implies a change in the number of instances delivering tuples/watermarks to $D$.
Independently of whether tuples/watermarks are delivered to $D$ through individual queues or $\X$ objects, a crucial question arises: \textit{can a newly provisioned instance deliver tuples/watermarks to a $d_l$ instance that conflict with $\watermarkof{d_l}$ (i.e., that carry a timestamp earlier than $\watermarkof{d_l}$)?}
Without assumptions on when such operator instance will deliver tuples/watermarks or which values they will carry, this is indeed the case.
\draft{Note this could violate correctness if $D$ is stateful and $t$ contributes to a window instance $w$ that $d_l$ has already treated as expired.}

One way, discussed in the literature~\cite{zacheilas2015elastic}, to deal with this is to relax the correctness guarantees. An alternative way, which we follow in this work, is to prove that, for an \generalopplusinmath{} operator, 
it is possible to guarantee a safe lower bound on the watermarks delivered by newly provisioned instances and thus consistently deliver tuples/watermarks to $D$ (this is shown in detail in Lemma~\ref{lem:addissafe}, later in this section).

\begin{figure}[ht!]
\includegraphics[width=\linewidth]{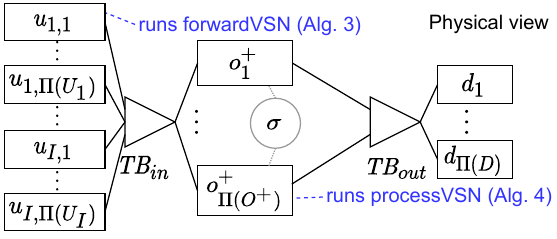}
\caption{\xxx{}'s model for VSN parallelism/elasticity.}\label{fig:vsn_elastic}
\end{figure}

\subsubsection*{Implementation of the \texttt{process} method}

\autoref{alg:o_vsn_elastic} overviews the implementation of \generalopplusinmath{} for VSN parallelism and elasticity.
It resembles the structure of the one in~\autoref{alg:o} but has important differences. 

The first difference between the algorithms is the additional L\ref{alg:o_vsn_elastic:adaptstart}-\ref{alg:o_vsn_elastic:adaptstop} to handle the elastic reconfigurations.
As shown, the algorithm first checks if the tuple is a control one (L\ref{alg:o_vsn_elastic:adaptstart}).
If that is the case, it proceeds storing $\gamma$ (the event time triggering the reconfiguration) and the future values of $e,\mathbb{O},\opmappingfun$ in the respective variables $e^*,\mathbb{O}^*,\opmappingfun^*$.
Notice that, as shown, a control tuple itself does not trigger a reconfiguration immediately. 
The reconfiguration is later triggered when a new incoming regular tuple increases the watermark and such increased watermark is greater than $\gamma$ (L\ref{alg:o_vsn_elastic:reconfigif}).
In this case, the algorithm enters a barrier and, once all instances reach such barrier, proceeds to apply the reconfiguration (L\ref{alg:o_vsn_elastic:barrier}-\ref{alg:o_vsn_elastic:adaptstop}) before processing $t$ (L\ref{alg:o_vsn_elastic:outstart} onward).

The second difference between the algorithms is that the code handling \generalopplusinmath{}'s state when producing output tuples (L\ref{alg:o_vsn_elastic:outstart}-\ref{alg:o_vsn_elastic:outend}). In this case, each instance $\generalopplusinst_j$ handles an expired window instance $w$ only if $w.k$ is $\generalopplusinst_j$'s responsibility (L\ref{alg:o_vsn_elastic:checkearlieast}). Together with method \texttt{handleInputTuple}, which handles only the keys responsibility of $\generalopplusinst_j$, this prevents concurrent modifications for the same key in $\opstate$.

\subsubsection*{Formal Guarantees}

Given~\autoref{alg:o_vsn_elastic}, we can now formally prove that \xxx{} ensures that each key maintained in the shared state $\opstate$ is consistently updated by exactly one instance at a time.


\begin{algorithm*}[h]
\footnotesize
\SetAlgoLined
\DontPrintSemicolon

\SetKwInput{KwParametersL}{Instance-local variables (besides $\WA,\WS,I,\keybymulti,\psipar,S,\opmapping,f_U,f_O,f_S$)}
\SetKwInput{KwParametersS}{Shared variables}
\SetKwInput{KwAuxiliaryFunctions}{Auxiliary functions}
\SetKwProg{Proc}{Function}{}{}

\KwParametersL{}
$W$ \tcp{$\generalopplusinst_j$'s watermark, initially 0}
$\rho$ \tcp{earliest $w.l$ of any $w \in \sigma$, initially 0}
$e,e^*$ \tcp{current/next epoch number}
$\mathbb{O},\mathbb{O}^*$ \tcp{set of current/next epoch instances}
$\opmappingfun^*$ \tcp{next epoch \opmappingfuninmath{}}
$\gamma$ \tcp{Event time to trigger reconfiguration}
\KwParametersS{}
$\sigma$ \tcp{shared state $\sigma$}

\BlankLine
\KwAuxiliaryFunctions{}
\BlankLine
 \tikzmark{squareanchor1}\FuncSty{update$W$($t$), }
\FuncSty{$\sigma$.remove($k,\ell$), } 
\FuncSty{$\sigma$.set($k,\ell,\{\zeta_1,\ldots,\zeta_I\}$), } 
\FuncSty{$\sigma$.shift($k,\ell,\{\zeta_1,\ldots,\zeta_I\}$), } 
\FuncSty{$\sigma$.check\&create($k,l$), } 
\FuncSty{earliestWinL($t$), } 
\FuncSty{latestWinL($t$), }
\FuncSty{prepareOutTuples($\{\varphi^1,\ldots,\varphi^\ell\}$), }
\FuncSty{forwardAndShift($k$), } 
\tikzmark{squareanchor2}\FuncSty{handleInputTuple($t$)}\;
\BlankLine
\FuncSty{waitForInstances($\mathbb{O}$)} \tcp{blocking call acting as a barrier} 
\FuncSty{isControl($t$)} \tcp{check if $t$ is a control tuple}
\FuncSty{prepareReconfig($t$)} \tcp{setup reconfig-related parameters}
\BlankLine

\Proc{\FuncSty{processVSN}($t$)}
{
    \lIf(\tcp*[h]{set up reconfiguration parameters. The reconfiguration is triggered as soon as $W$ grows beyond $\gamma$}){\FuncSty{isControl($t$)}}{\label{alg:o_vsn_elastic:adaptstart}
         $\{e^*,\mathbb{O}^*,\opmappingfun^*,\gamma\} \xleftarrow{} $\FuncSty{prepareReconfig}($t$)\label{alg:o_vsn_elastic:prep}
    } \Else {
    $\overline{W} \xleftarrow{} W$\; \label{alg:o_vsn_elastic:wold}
    \FuncSty{update$W$($t$)} \tcp{update $\generalopplusinst_j$ watermark}\label{alg:o_vsn_elastic:water}
    \If{$W>\overline{W} \wedge W > \gamma$}{\label{alg:o_vsn_elastic:reconfigif}
        \FuncSty{waitForInstances($\mathbb{O}$)}\;\label{alg:o_vsn_elastic:barrier}
        \lIf(\tcp*[h]{provision instances. The \FuncSty{if} clause ensures exactly one instance adds new instances first to $\X_{out}$ and then to $\X_{in}$}){$|\mathbb{O}^*|>|\mathbb{O}| \wedge \Xplus_{out}$.\texttt{addSources($\mathbb{O}^*\setminus \mathbb{O}$)}}{
            $\Xplus_{in}$.\texttt{addReaders($\mathbb{O}^*\setminus \mathbb{O},j$)}}\label{alg:o_vsn_elastic:addSources}
        \lIf(\tcp*[h]{decommission insts. The \FuncSty{if} clause ensures exactly one instance removes instances first from $\X_{in}$ and then from $\X_{out}$}){$|\mathbb{O}^*|<|\mathbb{O}| \wedge \Xplus_{in}$.\texttt{removeReaders($\mathbb{O}\setminus \mathbb{O}^*$)}}{
            $\Xplus_{out}$.\texttt{removeSources($\mathbb{O}\setminus \mathbb{O}^*$)}}\label{alg:o_vsn_elastic:rmSources}
        $\{e,\mathbb{O},\opmapping\} \xleftarrow{} \{e^*,\mathbb{O}^*,\opmapping^*\}$\;\label{alg:o_vsn_elastic:adaptstop}
    }
    \While(\tcp*[h]{while $\generalopplusinst_j$ has expired window inst. it is responsible for and starting at $\rho$}){$\rho+\WS<W$}{\label{alg:o_vsn_elastic:outstart}
        \lWhile{$\exists k | \sigma[k][1].l = \rho \wedge \opmapping(k)=j$}{\label{alg:o_vsn_elastic:checkearlieast}
            \FuncSty{forwardAndShift}($k$)
        }
        $\rho\xleftarrow{}\rho+\WA$ \tcp{update $\rho$ value}\label{alg:o_vsn_elastic:outend}
    }
    \begin{tikzpicture}[remember picture, overlay]
    \node (A1) at ([]pic cs:squareanchor1) {};
    \node (B1) at ([]pic cs:squareanchor2) {};
    \draw[draw=black]  let 
    \p1 = (A1), 
    \p2 = (B1) in (\x1-2,\y1+8) rectangle (\x2+100,\y2-5);
    \path let 
    \p1 = (B1) in node[blue,left]  at (\x1+100,\y1+22) {As defined in \autoref{alg:o}, but operating on $\sigma$ rather than $\sigma_j$};
    \end{tikzpicture}
    \FuncSty{handleInputTuple}($t$)\label{alg:o_vsn_elastic:handlein}
    }
}
\caption{Method \texttt{processVSN} (\texttt{process} for VSN setups) of a $\generalopplusinst_j$ instance. Runs when $t$ is returned by $\X_{in}$.}
\label{alg:o_vsn_elastic}
\end{algorithm*}

\begin{theorem}
\label{thm:vsndynamic}
If $\mathbb{O} \cup \mathbb{O}^*$ instances run the \texttt{processVSN} method presented in~\autoref{alg:o_vsn_elastic}, then elastic reconfigurations can be carried out while preserving \generalopplusinmath{}'s semantics (cf.~\autoref{ssc:stretchgeneralizedmodel}).
\end{theorem}

\begin{proof}[Proof]
\draft{To begin, note that when processing each tuple, all instances in $\mathbb{O}$ use the same \opmappingfuninmath{} (L\ref{alg:o_vsn_elastic:adaptstart}-\ref{alg:o_vsn_elastic:adaptstop}).
Hence, each key $k$ is only updated by the instance responsible for $k$.}

We argue that each key $k$ is also updated exclusively and consistently by one instance in the presence of reconfigurations.
Each instance in $\mathbb{O}$ enters the \texttt{if} statement at L\ref{alg:o_vsn_elastic:reconfigif} only after it receives a tuple $t$ that increases its watermark to the first value greater than $\gamma$ and then proceeds to wait for all other instances (L\ref{alg:o_vsn_elastic:barrier}).
As specified in Definition~\ref{def:xin}, $\X$ objects deliver the same watermarks to all instances, and each watermark observed by an instance has non-decreasing values.
Before entering the barrier, all instances have already handled expired window instances whose right boundary fell before $\overline{W}$ using $\opmappingfun$.
Once leaving the barrier, all instances will consistently handle window instances whose right boundary falls after $\overline{W}$ using $\opmappingfun^*$ (both expired and non-expired ones).
Newly provisioned instances (if any) are connected to $\X_{in}/\X_{out}$, or alternatively instances being decommissioned (if any) are disconnected from $\X_{in}/\X_{out}$ by exactly one of the existing instances (the one that succeeds in adding the sources being provisioned or in removing the readers being decommissioned).
Hence, if $t$ carries a key $k$ whose responsibility has shifted from $\generalopplusinst_i$ to $\generalopplusinst_j$, independently of whether $\generalopplusinst_j$ is a newly deployed instance or not, $\generalopplusinst_j$ will not only process $t$ (in relation to key $k$) after $\generalopplusinst_i$ is done processing all tuples preceding $t$, but also process $t$ after having been connected to both $\Xplus_{in}$ and $\Xplus_{out}$, if $\generalopplusinst_j$ is a newly deployed instance.
If $\generalopplusinst_i$ is being decommissioned, $\generalopplusinst_i$ will not process $t$ (no key responsibility of $\generalopplusinst_i$ is returned by $\opmappingfun^*$) and $\generalopplusinst_i$ will also not retrieve any tuple 
from $\Xplus_{in}$ nor will it output to $\Xplus_{out}$ once disconnected from the latter.
\end{proof}

After proving~\autoref{alg:o_vsn_elastic} can support parallel execution and elastic reconfigurations for \generalopplusinmath{} while enforcing \generalopplusinmath{}'s semantics correctly, we now prove it also enables VSN parallelism/elasticity for any downstream peer. This is because, even in the presence of provisioning reconfigurations, \generalopplusinmath{} can consistently deliver tuples/watermarks to $\X_{out}$, and the latter can act as the $\X_{in}$ of $D$. 

\begin{lemma}
\label{lem:addissafe}
Being $t$ the tuple that triggers a reconfiguration (entering the \texttt{if} statement at L\ref{alg:o_vsn_elastic:reconfigif}) $t.\tau$ is a safe lower bound for the watermark of newly provisioned instances.
\end{lemma}

\begin{proof}[Proof]
The method \texttt{addSources} (L\ref{alg:o_vsn_elastic:addSources}) is invoked successfully by exactly one $\generalopplusinst_j$ instance only upon reception of a tuple $t$ that increases $\generalopplusinst_j$'s watermark (from $\overline{W}$ to $W$, L\ref{alg:o_vsn_elastic:wold}-\ref{alg:o_vsn_elastic:barrier}). 
All results that could have been produced before processing $t$ have already been delivered to $\Xplus$ by all instances (\autoref{thm:vsndynamic}) and have timestamp lower than or equal to $\overline{W}$ and thus lower than or equal to $t.\tau$ (since $t.\tau \geq \overline{W}$, cf. Definition~\ref{def:watermark}), while all results that could depend on $t$ or future tuples will have a timestamp greater than $t.\tau$ (cf.~Observation~\ref{obs:outtuplestimestamp}).
If according to Lemma~\ref{lem:deliveringoutputs}, the timestamps of the tuples produced by $\generalopplusinst_j$ can be used as watermarks by $\X$, then $t.\tau$ can be immediately delivered as a watermark for a newly provisioned instance $\generalopplusinst_j$.
\end{proof}

\separate
\section{Using ScaleGate as $\X$ Objects}
\label{sec:esg_impl}

\autoref{sec:stretch} covered \xxx{}'s algorithms for VSN parallelism and elasticity. The latter relies on the $\X$ data object, which exposes six methods, presented in Table~\ref{tab:xapi}.
Here, we show how an enhanced \scalegate{} object (cf.~\autoref{ssc:scalegate}) can offer all such methods and support a real implementation of \xxx{}.

We begin observing that the original \scalegate{} is already sufficient to provide methods \texttt{add} and \texttt{get}. 
More formally, under the assumption that each source delivers a timestamp-sorted stream of tuples, \scalegate{} objects allow concurrent insertion and retrieval of tuples for arbitrary sets of sources and readers, delivering each tuple exactly once, and also delivering non-decreasing implicit watermarks (cf.~\autoref{ssc:correctness}) through tuples' $t.\tau$ attribute.
As discussed in \autoref{sec:generalizedmodel}, each $\generalopplusinst_j$ outputs window results in timestamp-order (Lemma~\ref{lem:deliveringoutputs}).
It is thus safe for each $\generalopplusinst_j$ to deliver its output tuples to an \scalegate{} data structure that can support data-duplication-free parallelism for downstream peers too, in a composable fashion.
In this case, both \generalopplusinmath{} as well as $D$ (if the latter is a stateful operator) can support correct execution for both order-sensitive as well as order-insensitive functions.

\scalegate{} objects do not define methods to dynamically change their sources and readers. They can be extended to provide the API methods highlighted in \autoref{tab:xapi}, nonetheless. We refer to such extended objects as \textit{Elastic ScaleGate} (\elasticscalegate{}) objects.
In order to outline our \elasticscalegate{} implementation, we first overview the internals of the \texttt{add} and \texttt{get} methods.
\scalegate{} builds a skip list where tuples are maintained in timestamp order, along with some auxiliary book-keeping structures. 

\begin{figure}[h!]
\includegraphics[width=\linewidth]{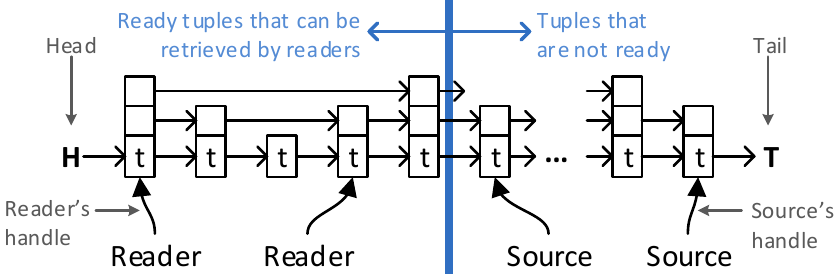}
\caption{ScaleGate's skip list, and readers'/sources' handles.}\label{fig:sgnutshell}
\end{figure}

The book-keeping structures contain handles to the skip list, for sources and readers, to continue inserting or reading nodes (tuples), respectively.
As shown in~\autoref{fig:sgnutshell}, readers' handles traverse the list from head to tail, retrieving the next tuple only if the latter is not pointed by a source's handle (thus returning only ready tuples). At the same time, sources' handles point to their last inserted tuples and facilitate the sorted insertion of subsequent tuples (also leveraging the skip list shortcuts).
Since each source adds a timestamp-sorted stream, each insertion ``falls'' after its previous one.
All the tuples before the earliest tuple pointed by a source (i.e., with earlier timestamps), are ready.

\noindent\textit{Adding new readers:}
Each reader has access to one of \elasticscalegate{}'s nodes through its own handle.
Each new reader simply needs a handle to the node pointed by the $j$-th reader (the one invoking the \texttt{addReaders} method) to retrieve next the same tuple that will be delivered to the $j$-th reader (according to the API), and then traverse the rest of the list in timestamp order in subsequent  \texttt{get} invocations.

\noindent\textit{Removing existing readers:}
Removing a set of existing readers only requires removing their thread-specific structures.

\noindent\textit{Adding new sources:}
According to Lemma~\ref{lem:addissafe}, the timestamp of the tuple $t$ triggering a reconfiguration is a safe lower bound for the watermark of newly provisioned instances in \xxx{}. Since $t$ is the last tuple observed by the instance invoking method \texttt{addSources}, and given that $t_o.\tau < \watermarkof{\generalopplusinst_j} < t.\tau$, being $t_o$ the last tuple produced by such $\generalopplusinst_j$ instance, the book-keeping handles of new sources can be copying the handles of the source invoking successfully the \texttt{addSources} method.
For the sake of the new thread, an initial \emph{dummy} tuple following the one pointed by the source invoking successfully the \texttt{addSources} method is inserted,  to initialize the functionality of its handles. Readers can move their pointer to the next tuple pointed by a source when such tuple is of type dummy but the tuple is not returned as ready to readers invoking \texttt{get}.

\noindent\textit{Removing existing sources:}
Removing a source consists mainly of adding, on behalf of the source, a special \textit{flush} tuple in \elasticscalegate{}, with a timestamp equal to the latest insertion of the source. Such a tuple will let the previously added tuples of the leaving source to be ready (according also to other sources' tuples). The removed source's associated book-keeping structures can then be removed. As for dummy tuples, readers can move their pointer to the next tuple pointed by a source when such tuple is of type flush, without returning it as ready via the \texttt{get} method.

\noindent\textit{Concurrent calls to the API methods:}
For concurrent calls of the same method that updates the set of threads (e.g. concurrent calls to \texttt{addReaders}), synchronization is in place (using a TestAndSet variable) so that only one of each type takes effect. Concurrent calls among competing such methods that modify the thread-specific book-keeping structures (e.g. \texttt{addReaders} and \texttt{removeReaders}) require extra synchronization to protect consistency; since these are low-contention operations, nonetheless, a simple lock can do.
Since each reconfiguration results in sources/readers being added or removed but not both (\autoref{alg:o_vsn_elastic} L\ref{alg:o_vsn_elastic:addSources}-\ref{alg:o_vsn_elastic:rmSources}), and each reconfiguration can only start after the previous is completed (cf. \autoref{sec:implementation}), we do not incur such extra synchronization overhead in our implementation.
If regular operations (\texttt{add} and \texttt{get}) are concurrent with those that update the set of threads and their book-keeping structures, the latter can overwrite, causing the former to have no effect. Note that in \xxx{} such invocations do not interfere. 

\separate
\section{Implementation - API and Architecture}
\label{sec:implementation}

We focus herein on how \xxx{} meets Cond. 1 and  Cond. 2 (cf.~\autoref{sec:stretch}).
We begin overviewing the overall architecture of our implementation and discussing Cond. 1.

\begin{figure}[ht!]
\includegraphics[width=\linewidth]{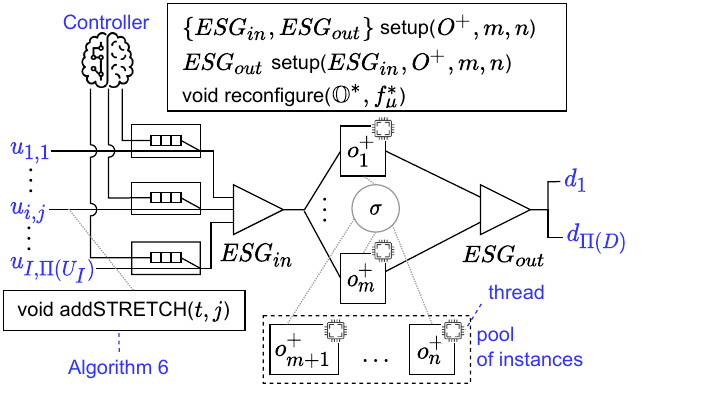}
\caption{API and architecture of \xxx{}'s implementation.}\label{fig:stretch}
\end{figure}

As shown in~\autoref{fig:stretch}, \xxx{}'s API defines two \texttt{setup} and one \texttt{reconfigure} method.
Method $\{ESG_{in},ESG_{out}\}=\texttt{setup}(\generalopplus,m,n)$ takes as input the \generalopplusinmath{} to instantiate, its initial parallelism degree $m$ and its maximum parallelism degree $n$.
Upon invocation of this method, \xxx{} creates $n$ $\generalopplusinst_j$ that share state \opstateinmath{}. Furthermore, it connects $m$ of them to \elasticscalegatein{} and \elasticscalegateout{}. The remaining $n-m$ instances are kept in a pool. The pool is also used to keep instances that are later decommissioned.

Each $\generalopplusinst_j$ instance is run by a thread.
The latter is constantly trying to get a tuple to process through the \texttt{get} method of \elasticscalegatein{}.
When no tuple is retrieved, either because no tuple is ready or because $\generalopplusinst_j$ belongs to the pool and is thus not connected to \elasticscalegatein{}, exponential backoff prevents the thread from creating contention on \elasticscalegatein{}.
Instances can promptly start retrieving tuples once provisioned and connected to \elasticscalegatein{}, thus invoking method \texttt{processVSN}, while decommissioned/pool instances will create negligible contention on \elasticscalegatein{}, and stop invoking the \texttt{processVSN}/\texttt{forwardVSN} methods, as per Cond. 1.

The \texttt{setup} method returns both \elasticscalegatein{} and \elasticscalegateout{}, for the latter to be fed by upstream and to feed downstream instances, respectively.
As mentioned in~\autoref{sec:smps}, we focus our discussions on a single stateful operator for simplicity, but \xxx{} can be used to instantiate many (connected) operators within a query. 
Because of this, the $\{ESG_{out}\}=\texttt{setup}(ESG_{in},\generalopplus,m,n)$ variation can be used to create an \generalopplusinmath{} operator and connect it to a previously created one through \elasticscalegatein{} (i.e., the \elasticscalegateout{} of such upstream peer).

\begin{algorithm}[h]
\setcounter{algocf}{4}
\footnotesize
\SetAlgoLined
\DontPrintSemicolon
\SetKwInput{KwParameters}{Instance-local variables}
\SetKwProg{Proc}{Function}{}{}
\KwParameters{}
$q[]$ \tcp{Queues holding reconfiguration messages}
$T[]$ \tcp{$\tau$ of the latest tuple added by source $i$}
\BlankLine
\Proc{\FuncSty{add\xxx{}}($t,i$)}
{
    $T[i]\xleftarrow{}t.\tau$\;\label{alg:modifiedass:lasttau}
    \While{$q[i]$.\FuncSty{size()}$>0$}{
        $x \xleftarrow{} q[i]$.\FuncSty{pop()}\;
        \elasticscalegatein{}.\FuncSty{add($\langle T[i],\text{control}, \left[ x.e, x.\mathbb{O}, x.\opmappingfun \right] \rangle$,$i$)}\;
    }
    \elasticscalegatein{}.\FuncSty{add($t$,$i$)}\;
}
\caption{Method \FuncSty{add\xxx{}}.}
\label{alg:modifiedadd}
\end{algorithm}

\begin{algorithm}[h]
\footnotesize
\SetAlgoLined
\DontPrintSemicolon
\SetKwProg{Proc}{Function}{}{}

\Proc{\FuncSty{prepareReconfig}($t$)}
{
    \If(\tcp*[h]{the reconfiguration id carried by $t$ is greater than \generalopplusinmath{}'s one}){$t.\varphi[1]>e$}{
        $e^* \xleftarrow{} t.\varphi[1]$ \tcp{set reconfig. parameters}
        $\mathbb{O}^* \xleftarrow{} t.\varphi[2]$\;
        $\opmappingfun^* \xleftarrow{} t.\varphi[3]$\;
        $\gamma \xleftarrow{} t.\tau$
    }
}
\caption{Method~\FuncSty{prepareReconfig}~used in~\autoref{alg:o_vsn_elastic}, to set the instance-local variables needed during a reconfiguration.}
\label{alg:preparereconfig}
\end{algorithm}

\begin{table*}[t]
\caption{\draft{Questions addressed in the evaluation, together with information about the operators used in the experiments.}}
\label{tab:evalq}
\begin{center}
\begin{tabular}{@{}lp{.6cm}p{7cm}p{7cm}l@{}}
\toprule
ID   & Setup & Question & Operator & Section                     \\ \midrule
$Q_1$ & Static & How do VSN (\xxx{}) and SN (Flink) compare for established baselines such as \textit{wordcount}? & $\aggplusop$ implementing \texttt{wordcount} and \texttt{paircount} (a variation that counts distinct pairs of words) & \autoref{ssc:q1} \\ 
$Q_2$ & Static & What is the maximum throughput/minimum latency for VSN/SN setups in \xxx{} and Flink? & Operator~\ref{ex:o2} (cf. Appendix~\ref{apx:examples}), processing tweets collected during October 2018. & \autoref{ssc:q2} \\ 
$Q_3$ & Static & How does \xxx{} compare with ad-hoc stateful operator implementations such as ScaleJoin (Operator~\ref{ex:j})? & ScaleJoin (Operator~\ref{ex:j}, cf. Appendix~\ref{apx:examples}) with $\WA$/$\WS$ set to $\delta$ and 5 minutes, resp., running the benchmark from~\cite{scalejoin}. & \autoref{ssc:q3} \\
$Q_4$ & Elastic & How long does it take for \xxx{} to complete an elastic reconfiguration? & ScaleJoin (Operator~\ref{ex:j}, cf. Appendix~\ref{apx:examples}) with $\WA$/$\WS$ set to $\delta$ and 5 minutes, resp. & \autoref{ssc:q4} \\
$Q_5$ & Elastic & What is \xxx{} performance under multiple reconfigurations?                         & ScaleJoin (Operator~\ref{ex:j}, cf. Appendix~\ref{apx:examples}) with $\WA$/$\WS$ set to $\delta$ and 1 minute, resp. & \autoref{ssc:q5} \\
$Q_6$ & Elastic & What is \xxx{} performance in real-word applications? & ScaleJoin (Operator~\ref{ex:j}, cf. Appendix~\ref{apx:examples}) with $\WA$/$\WS$ set to $\delta$ and 30 seconds, resp., analyzing financial trades & \autoref{ssc:q6} \\ \bottomrule
\end{tabular}
\end{center}
\end{table*}

\subsubsection*{Triggering elastic reconfigurations}

As discussed in~\autoref{sec:stretch}, elastic reconfigurations depend on the parameters $e^*$, $\mathbb{O}^*$, $\opmappingfun^*$, $\gamma$, which in \autoref{alg:o_vsn_elastic} are handled by methods \texttt{isControl} and \texttt{prepareReconfig}, detail next.

In \xxx{}, a reconfiguration is triggered by the external module sharing a new set of instances ids $\mathbb{O}^*$ and a mapping function $\opmappingfun^*$.
To deliver $\mathbb{O}^*$ and $\opmappingfun^*$ to \generalopplusinmath{} instances, \xxx{} encapsulates them in a special control tuple.
Method \texttt{isControl} distinguishes regular from control tuples based on the attributes carried in their metadata.

It should be noted that, even if regular as well as control tuples can be delivered by $U_{i,j}$ instances and the controller, respectively, \elasticscalegatein{} still expects each of its sources to deliver tuples in timestamp order. To fulfill such a condition,
\xxx{} defines one dedicated control queue per upstream instance (as shown in~\autoref{fig:stretch}), and wraps the \texttt{add} method of \elasticscalegatein{} within the method \texttt{addSTRETCH}, shown in \autoref{alg:modifiedadd}.
With this method, \xxx{} tracks the last timestamp $\tau$ forwarded by each upstream instance (\autoref{alg:modifiedadd} L\ref{alg:modifiedass:lasttau}) and, by having the \texttt{reconfigure} method add the next $\mathbb{O}^*,\opmappingfun^*$ in each control queue, creates and forwards a control tuple carrying timestamp $\tau$ and $\mathbb{O}^*,\opmappingfun^*$ in its metadata.
Control tuples can then be processed as shown in \autoref{alg:preparereconfig}. 

\subsubsection*{Formal Guarantees}

\begin{theorem}
\label{thm:atomicity}
Each reconfiguration that is applied based on~\autoref{alg:o_vsn_elastic}, \autoref{alg:modifiedadd}, and \autoref{alg:preparereconfig} takes place atomically and exactly once.
\end{theorem}

\begin{proof}[Proof]
As shown in~\autoref{alg:o_vsn_elastic}, all $\generalopplusinst_j$ perform an elastic reconfiguration based on their local parameters $e^*$, $\mathbb{O}^*$, $\opmappingfun^*$, and $\gamma$, with exactly one instance succeeding in connecting or disconnecting provisioned or decommissioned instances from $\X_{in}$/$\X_{out}$, respectively. 
These parameters are set by~\autoref{alg:preparereconfig} based on a control tuple $t$, delivered by \elasticscalegatein{} to all $\generalopplusinst_j$. Furthermore, if multiple such control tuples are delivered by \elasticscalegatein{}, all are delivered in the same order to $\generalopplusinst_j$ instances. Hence, all instances switch to the same $e^*$ at the same point in time. If multiple $e^*$ are delivered by multiple control tuples from \elasticscalegatein{}, the latest one is applied, and such latest one is the same for all $\generalopplusinst_j$ instances.
\end{proof}

\separate
\section{Evaluation}
\label{sec:eval}

We aim at comparing \xxx{} with other state-of-the-art baselines.
\draft{We place emphasis on stream joins since joins are among the most computationally heavy operators~\cite{ananthanarayanan2013photon}.
The choice of focusing on stream joins is also motivated by the existence of~\cite{scalejoin}, a custom highly-specialized implementation of the VSN parallelism offered by \xxx{} that, while supporting only a fraction of the stateful analysis that \xxx{} can support, provides the best performance figures at which \xxx{} can aspire.
To account also for other state-of-the-art baselines, we also compare with Apache Flink (or simply Flink).
First, we rely on an established baseline, word-count~\cite{miao2017streambox} and a variation counting distinct \textit{pairs} of words, to study the effects of different data duplication levels on throughput and latency metrics. We then study the maximum performance \xxx{} and Flink can offer for \generalopplusinmath{} operators with $I=2$ (i.e., including joins).
Since both ScaleJoin and \xxx{} support correct execution for order-sensitive functions, we assume that in SN setups input tuples are merged-sorted by both $\generalopplusinst_j$ and $d_j$ instances.
Our experiments seek answers to the questions found in~\autoref{tab:evalq}: $Q_1$-$Q_3$ assume a static setup, while $Q_4$-$Q_6$ focus on elastic setups and also on real-world applications.}

Experiments are run on a 2.10 GHz Intel(R) Xeon(R) E5-2695 CPU with 2 18-cores sockets, 72 logical threads with hyper-threading, and 64 GB memory. 
\xxx{} is implemented in Java and tested with Java HotSpot(TM) 64-bit Server VM.
For SN, we use Flink 1.6.0.

In the following, we present and discuss the results of the performance metrics of interest,
averaged over five runs. 
More concretely, we use input rate -- computed as the number of tuples/second (t/s) processed by an operator, throughput (for join operators) -- computed as the number of comparisons/second (c/s) sustained by the operator, and latency -- computed as the timestamp difference of each output tuple and the latest input tuple that contributes to it, 
while using \textit{flow control} to handle backpressure.
The implemented flow control mechanism is similar to that of Flink, 
in this case, putting a bound on \elasticscalegate{}'s size. 
For experiments concerning elasticity, we also report reconfiguration times -- measured as the time difference between the moment the controller invokes method \texttt{reconfigure} (cf.~\autoref{sec:implementation}) to the moment the reconfiguration is completed, and the number of threads used throughout the experiment. 

\subsection{\draft{Throughput and latency in VSN (\xxx{}) vs SN (Flink) for established baselines such as \textit{wordcount} ($Q_1$)}}
\label{ssc:q1}

\draft{
In our first experiment, we use the \texttt{wordcount} benchmark~\cite{miao2017streambox}, frequently used in applications like Sentiment Analysis ones~\cite{qian2013timestream}.
To account for different amounts of duplication, we also run a \texttt{paircount} variation, that counts pairs of rather than individual words.
In these experiments, \xxx{}'s shared-memory allows each input tuple to be shared with all the parallel threads, having each one responsible for one word/pair, while Flink's shared-nothing approach requires each tuple to be transformed in multiple output tuples, according to Corollary~\ref{cor:duplicationviamap}. 
The definitions for all the used operators are found in Appendix~\ref{apx:examples}. 
}

\draft{
We process a dataset consisting of 4.3 million tweets, between the 1$^{st}$ and 2$^{nd}$ of October 2018.
The input schema is defined as $\langle \tau, \left[ \text{user}, \text{tweet} \right] \rangle$.
With \texttt{wordcount}, each input tuple $t$ results in as many output tuples as words carried in $t.\varphi[2]$.
With \texttt{paircount}, each word in $t.\varphi[2]$ is paired and forwarded with its nearby words, up to a distance of 3, 10, and $+\infty$ for the Low (L), Medium (M), and High (H) duplication cases, respectively.
\texttt{Wordcount} gives the least amount of duplication, \texttt{paircount} H the greatest.
}

\draft{
Results are presented in~\autoref{fig:q0}. Flink performance is shown as a shaded region since, according to Corollary~\ref{cor:duplicationviamap}, a Map $\mapop{}$ is required to split each input tuple into distinct words/pairs, and such $\mapop{}$ can also be executed in parallel. The bottom and upper lines represent the minimum and maximum throughput/latency observed for any parallelism degree of $\mapop{}$ in $[1,36]$. As shown, in the \texttt{wordcount} case (the one with the least amount of duplication) \xxx{} and Flink are comparable. Despite its degradation (because of the contention on shared resources) after reaching its peak throughput, \xxx{} is nonetheless able to achieve $+17\%$ throughput and $-94\%$ latency. \xxx{}'s benefits become even more evident for the \texttt{paircount} cases, achieving $+137\%$, $+237\%$, and $+283\%$ throughput and $-89\%$, $-94\%$, and $-94\%$ latency for L, M and H, respectively. Note in this case the throughput is decreasing for an increasing duplication level (e.g., from counting words to counting pairs) since each input tuple (carrying a tweet) results in a higher number of pairs, thus resulting in a higher workload.
}

\begin{figure}[t]
    \centering
	    \includegraphics[width=\linewidth]{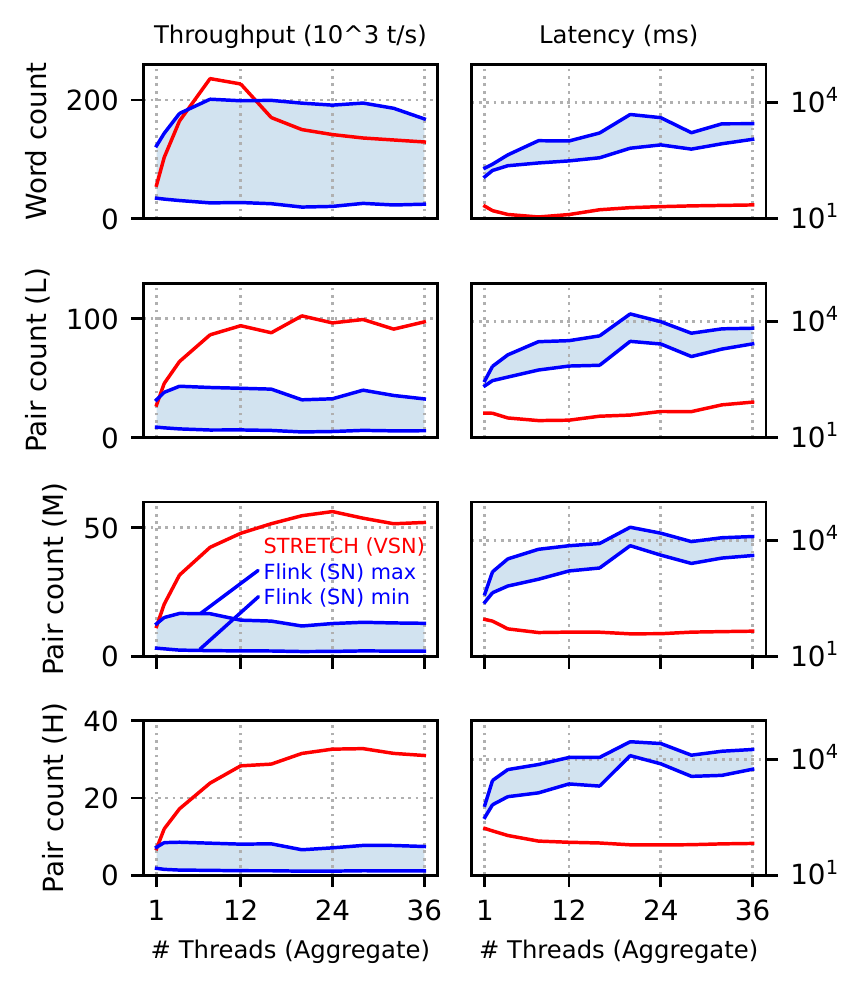}
    \caption{\draft{\xxx{} (VSN) and Flink (SN) Throughput/Latency for \texttt{wordcount} and \texttt{paircount}}} \label{fig:q0}
\end{figure}

\subsection{Maximum throughput and minimum latency in VSN (\xxx{}) vs SN (Flink) for an \generalopplusinmath{} with $I=2$ ($Q_2$)}
\label{ssc:q2}
\draft{This experiment compares \xxx{} performance using Flink as a baseline for SN operators that, like Joins, define two input streams.
Since Flink does not offer parallel execution of general Joins (but only EquiJoins), we evaluate the performance of an operator that simply forwards each incoming tuple, and thus measures the maximum throughput/minimum latency observed when the performance bottleneck is given by data sharing and sorting (its definition is found in Appendix~\ref{apx:examples}). 
We use the same dataset from $Q_1$.
}

\autoref{fig:modelComparison:twitter} shows the operators' scalability.
For an increasing $\Pi(\generalopplus)$, 
\xxx{}'s throughput decreases from approximately 120\,000 to 100\,000 t/s
due to the synchronization overheads induced by the higher $\Pi(\generalopplus)$.
Flink starts from a lower throughput and decreases faster,
from 40\,000 to 2\,000 t/s.
\xxx{} achieves from 3$\times$ to 50$\times$ better throughput.
Flink's latency, regardless of $\Pi(\generalopplus)$, is 
higher than 100 ms, while \xxx{}'s is always less than 30 ms.

\begin{figure}[t]
    \centering
	    \includegraphics[width=\linewidth]{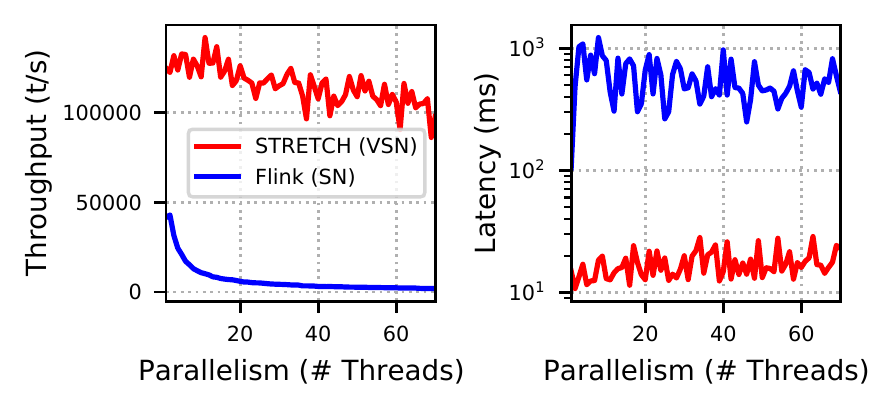}
    \caption{Max Throughput/min latency of \xxx{} (VSN) and Flink (SN) for a generic \generalopplusinmath{} with $I=2$ (Operator~\ref{ex:o2}).} \label{fig:modelComparison:twitter}
\end{figure}

\subsection{Throughput/latency of \joinplusopinmath{} in \xxx{} vs ad-hoc implementations ($Q_3$)}
\label{ssc:q3}

After studying Operator~\ref{ex:o2}'s maximum throughput/minimum latency, we focus now on \joinplusopinmath{} (Operator~\ref{ex:j}).
Since Flink supports only parallel equijoins, we compare \xxx{} performance with that of the original ScaleJoin implementation~\cite{scalejoin} and with an additional optimized single-threaded implementation, which we refer to as 1T.
The latter allows measuring the performance of an implementation that devotes as many CPU cycles as possible to data analysis rather than data communication when $\Pi(\joinplusop)=1$.

We follow the same benchmark used in ~\cite{scalejoin,roy2014low} to join two logical streams L and R.
L tuples' schema is $\langle \tau, \left[ x, y \right] \rangle$, where $x$ is type of \texttt{int} and $y$ is \texttt{float}.
R tuples' schema is $\langle \tau, \left[ a, b, c, d \right] \rangle$, where $a$, $b$, $c$, and $d$ are of type \texttt{int}, \texttt{float}, \texttt{double} and \texttt{boolean}, respectively.
$\WA$ and $\WS$ are set to $\delta$ (1 ms, as in Flink) and 5 minutes, respectively.
For each pair of tuples $t_L$ and $t_R$, an output tuple 
is produced if:
\begin{multline*}
t_R.\varphi[1] - 10 \leqslant t_L.\varphi[1] \leqslant t_R.\varphi[1] + 10\ \bigwedge \\ t_R.\varphi[2] - 10 \leqslant t_L.\varphi[2] \leqslant t_R.\varphi[2] + 10
\end{multline*}
Attributes $x$, $y$, $a$ and $b$ are randomly selected from a uniform distribution with interval $[1, 10\,000]$ which results on average in an output tuple every 250\,000 comparisons.

\autoref{fig:joinScalability} shows (i)~the average of the maximum sustainable input rate and the standard deviation for increasing $\Pi(\joinplusop)$, (ii)~the corresponding throughput, in terms of number of comparisons, and (iii)~the corresponding latency in logarithmic scale.
The throughput of \xxx{} and ScaleJoin when $\Pi(\joinplusop)=1$ is similar to 1T, but the latency for 1T is lower than the other two, because of \xxx{}'s and ScaleJoin's data communication costs.
\draft{As shown, despite being a general rather than a specialized operator, \xxx{}'s throughput still grows linearly with $\Pi(\joinplusop)$ and can match that of ScaleJoin (the latter shows some higher degradation caused by hyper-threading when $\Pi(\joinplusop)$ exceeds the 36 physical threads), despite a small overhead (in the order of 10 ms) in latency.}

\begin{figure}[t!]
	\centering
	\begin{subfigure}{.9\linewidth}
		\includegraphics[width=\textwidth]{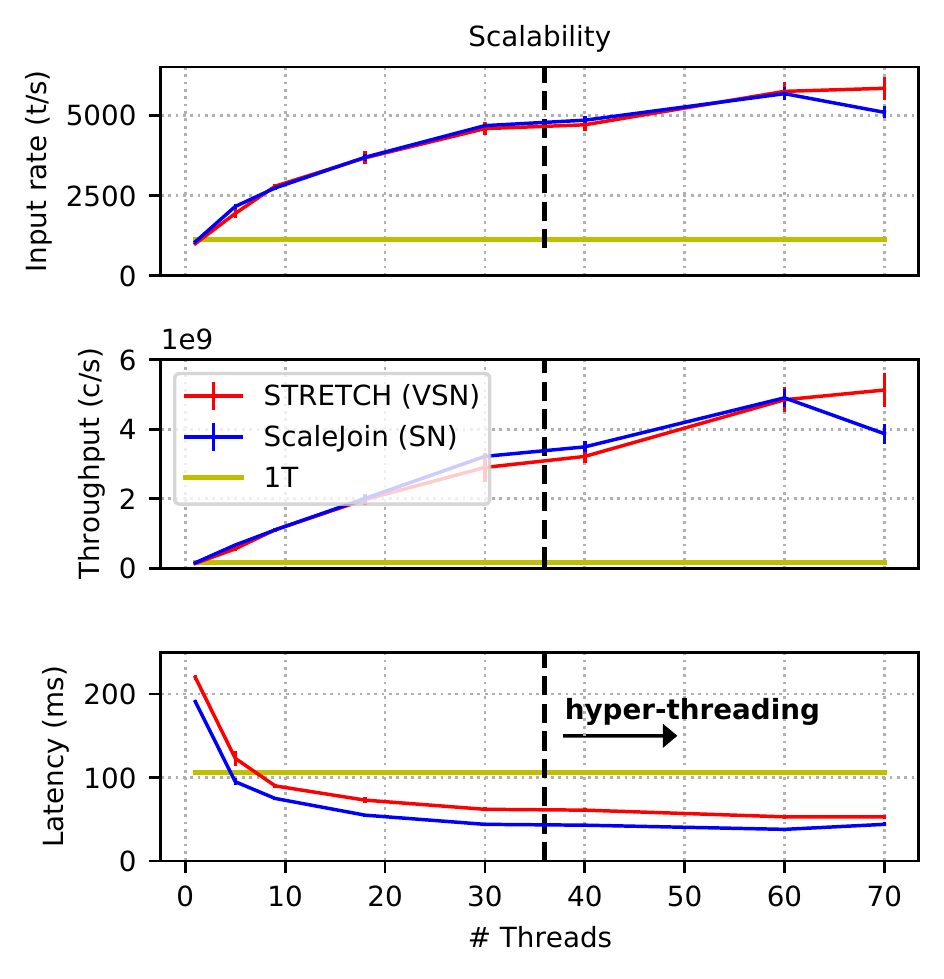}
	\end{subfigure}
	\caption{\small The performance of \xxx{} compared with those of ScaleJoin (SN) and 1T for \joinplusopinmath{}.}
	\label{fig:joinScalability}
\end{figure}

\subsection{\xxx{} -- elasticity overheads ($Q_4$)}
\label{ssc:q4}

We now move from static to elastic setups for the ScaleJoin benchmark (cf.~\autoref{ssc:q3}).
With $Q_4$, we measure the overhead and duration of individual elastic reconfigurations in isolation. 
We trigger one elastic reconfiguration per experiment using a simple controller.
The latter defines an upper, a target, and a lower CPU consumption threshold of 90\%, 70\%, and 45\%, respectively.
When the current processing load of active threads exceeds the upper threshold, the smallest amount of new threads needed to bring the average processing capacity below the target threshold is provisioned. 
When the current processing load of active threads is below the bottom threshold, the largest amount of underutilized threads that keep the average processing capacity below the target threshold is decommissioned.

To evaluate the elasticity of the framework, we increase (decrease) the load after filling the window and add (remove) threads while measuring the latency, throughput, and reconfiguration time.
For the provisioning experiments, we start with an input rate set to 70\% of the maximum rate sustainable by the corresponding number of $\joinplusop$ threads.
After 6 minutes, when the window is full and the system is stable, we increase the rate to 120\% of the maximum sustainable rate and therefore trigger an epoch switch.
For the decommissioning experiments, we start with 70\% of the maximum sustainable rate and, after 6 minutes, decrease the rate to 30\%.
\draft{Table~\ref{tab:elastic_threads} summarizes the various provisioning/decommissioning experiments. For each number of starting threads, it shows the resulting number of threads after provisioning/decommissioning or a ``-" when the corresponding number of starting threads does not allow for a provisioning/decommissioning action (e.g., it shows ``-" in post-decommissioning for 1 starting thread since no thread can be decommissioned in such a case).
The reconfiguration times for each action are shown on the left side of~\autoref{fig:reconfigtimes}. Each reconfiguration time is reported on the X axis based on the number of threads before the reconfiguration (e.g., the provisioning point, shaped as a circle, on 30 threads indicates it took approximately 12 ms to switch from 30 to 52 threads).
To stress that each reconfiguration is triggered for a set of threads in which no thread has more spare resources than others (i.e., when the load is balanced), we report on the right-side of~\autoref{fig:reconfigtimes} the coefficient of variation percentage of threads' workload (in this case one point for each number of starting threads before a provisioning or a decommissioning reconfiguration).
All in all, all reconfiguration times are negligible, always lower than 40 ms (an higher reconfiguration time is observed when $\Pi(\joinplusop)$ grows beyond the number of threads of a single socket), and there is at most 2\% of load imbalance among threads.
}

\begin{table}[h]
    \scriptsize
    \centering
    \caption{\draft{$\Pi(\joinplusop)$ values for~\autoref{fig:reconfigtimes}}}
    \label{tab:elastic_threads}
    \begin{tabular}{@{}lcccccccc@{}}
        \toprule
        starting $\Pi(\joinplusop)$ & 1 & 5 & 9 & 18 & 30 & 40 & 60 & 70\\
        \midrule
        $\Pi(\joinplusop)$ post-provisioning & 2 & 9 & 16 & 31 & 52 & 69 & - & - \\
        $\Pi(\joinplusop)$ post-decommissioning & - & 2 & 3 & 7 & 12 & 17 & 25 & 30 \\
        \bottomrule
    \end{tabular}
\end{table}

\begin{figure}[h]
	\centering
	\includegraphics[width=\linewidth]{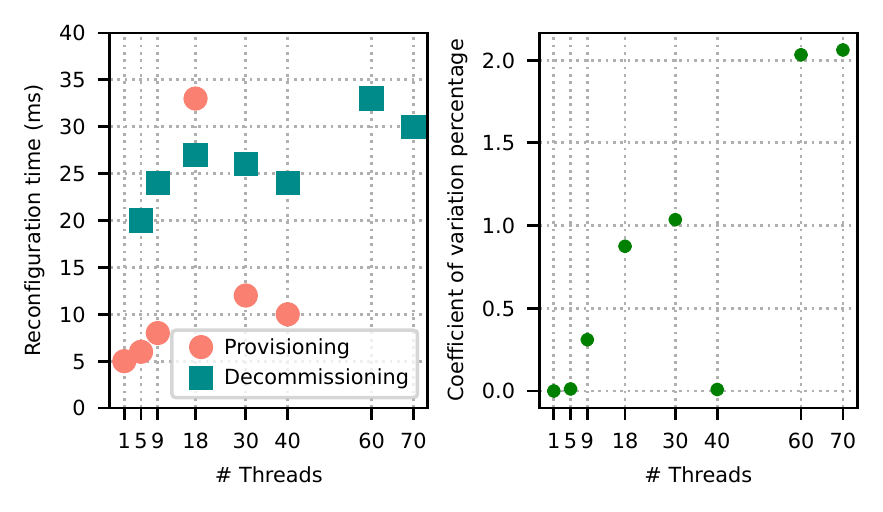}
	\caption{\draft{Reconfiguration times and coefficients of variation for provisioning/decommissioning elastic reconfigurations.}}
	\label{fig:reconfigtimes}
\end{figure}

\autoref{fig:provdecexample} shows the effect of increasing or decreasing the workload for \xxx{} and ScaleJoin, and consecutively provisioning or decommissioning threads for \joinplusopinmath{} in \xxx{}, when $\Pi(\joinplusop)$ is initially set to 18. As shown for the provisioning, \xxx{} sustains higher rates when adding new threads, resulting in higher throughput without affecting the latency.
In the decommissioning procedure, when decreasing the workload, \xxx{} achieves the same throughput as ScaleJoin but shows lower latency, as expected based on the model in~\cite{gulisano2017performance}.

\begin{figure}[t!]
	\includegraphics[width=\linewidth]{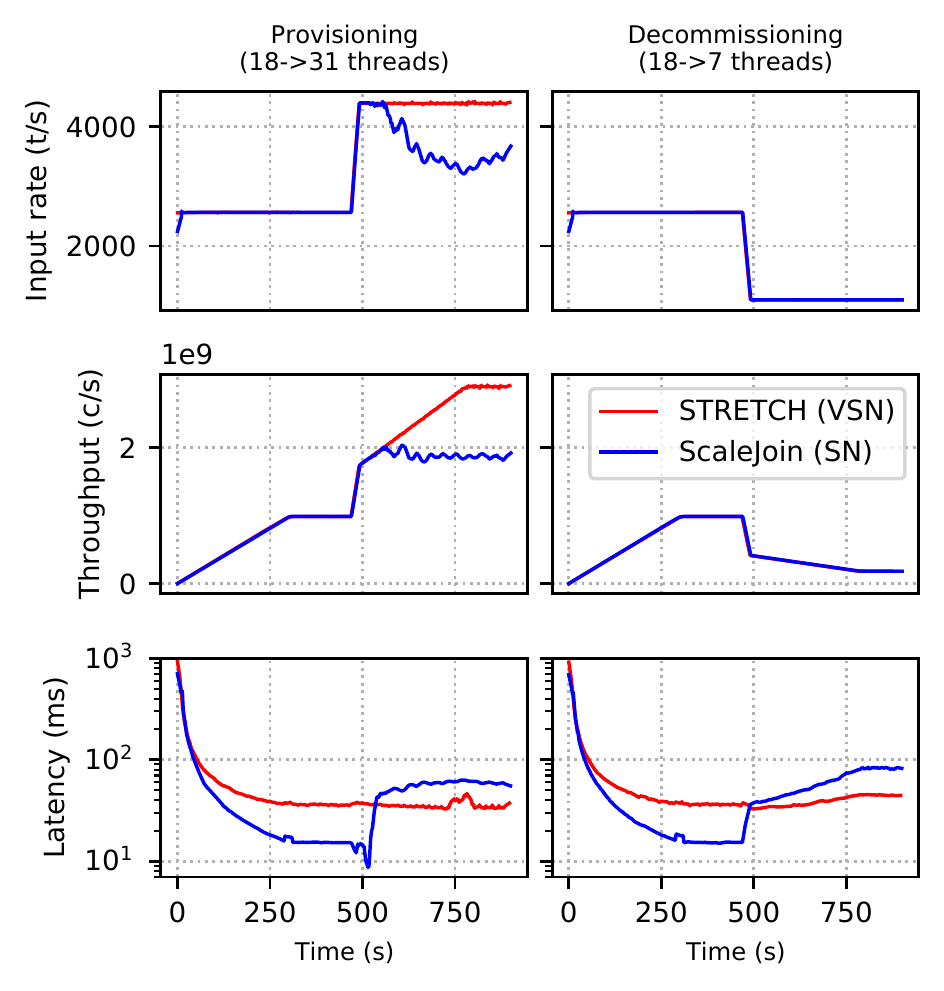}
	\caption{Input rate, throughput, and latency behavior during one provisioning/decommissioning elastic reconfiguration.}
	\label{fig:provdecexample}	
\end{figure}

\subsection{\xxx{} under multiple reconfigurations ($Q_5$)}
\label{ssc:q5}

We continue our evaluation with an experiment that aims at stress-testing \xxx{}'s reconfigurations.
First, we narrow the lower/upper thresholds to $[70\%,80\%]$ of the capacity.
Second, we modify the controller so that \joinplusopinmath{} cost is not matched only with its current CPU consumption, but accounts also for the pending and predicted workload, according to the model in~\cite{gulisano2017performance}.
Third, we decrease $\WS$ to one minute.
These changes imply more reconfigurations being frequently triggered (because of the narrower distance of lower/upper thresholds and because the controller proactively triggers reconfigurations on the predicted processing capacity in correspondence to rate changes) and higher sustainable input rates (because of the smaller $\WS$ and thus resulting comparisons per-tuple across windows).

Each experiment is 20 minutes long and consists of several sequential phases in which data tuples are injected with a constant rate, randomly chosen from the range $[500, 8\,000]$ t/s.
The length of each phase is at least 100 and at most 300 seconds (i.e., long enough for \joinplusopinmath{}'s window instances to reach the size corresponding to the phase's rate). 
The transition between phases is by an abrupt change in the input rate to stress timely reconfiguration. 
\autoref{fig:rateSeq} shows the measurements of a representative execution out of the ones conducted, to keep the illustration uncluttered. Results of additional experiments are shown in Appendix~\ref{apx:addexperiments}.

\autoref{fig:rateSeq}(a) shows the input rate throughout the experiment. 
As shown in \autoref{fig:rateSeq}(b), when the input rate increases, new processing threads are added accordingly. Likewise, in case of a decrease in the rate, unnecessary threads are removed without affecting throughput and latency.
\joinplusopinmath{}'s throughput is shown in~\autoref{fig:rateSeq}(c). 
The highlighted region indicates the upper and lower bound of computational capacity at each moment for the corresponding number of processing threads. 
When the workload exceeds the upper bound or falls short of the lower bound, the controller proceeds in applying reconfigurations.
\autoref{fig:rateSeq}(d) shows the average latency per second. 
There are a few spikes in the latency which are associated with abrupt changes in the input rate. However, as shown by the box plot on the right side, the controller keeps the average latency around 20 ms.
Also, as shown in the zoomed-in view, spikes are handled within 10 seconds.

\begin{figure}[ht!]
  \centering
  \includegraphics[width=\linewidth]{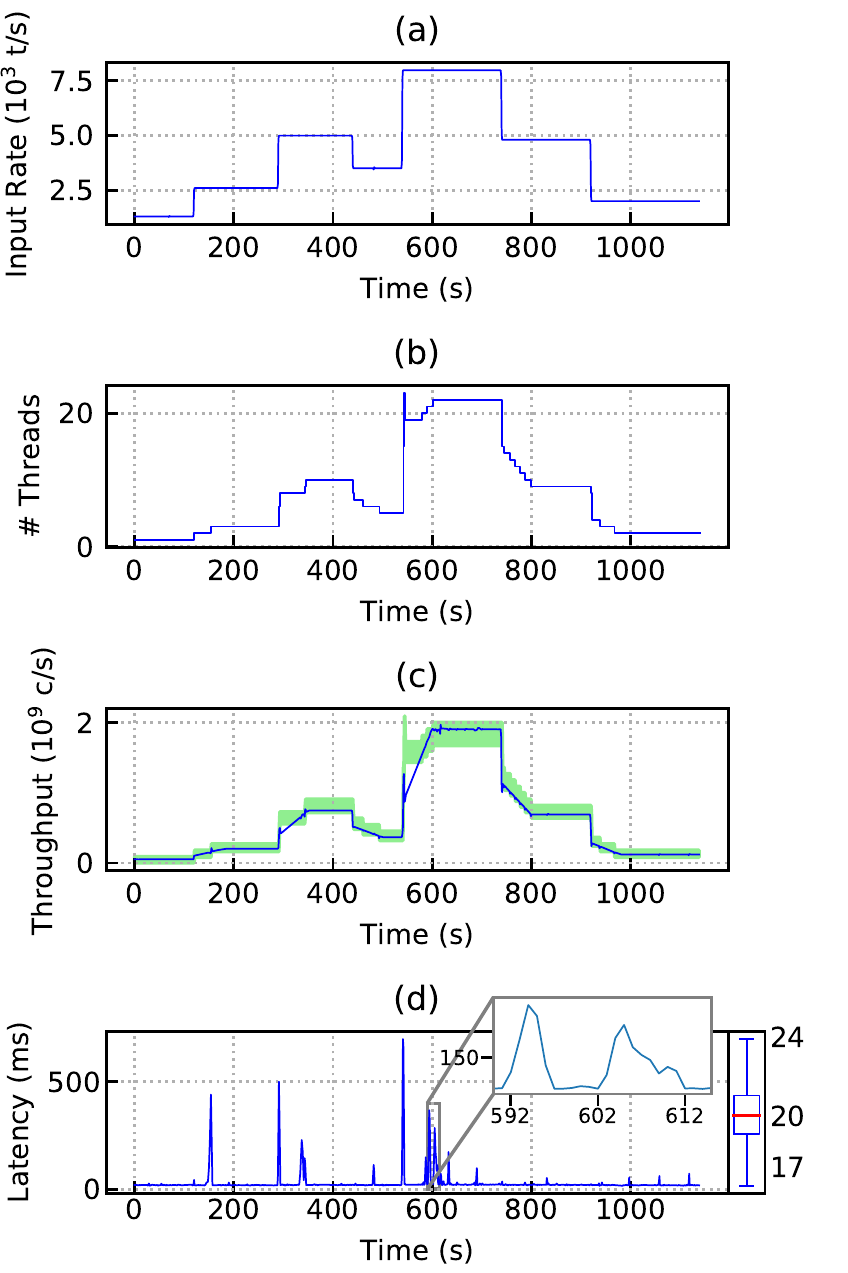}
  \caption{Results of adjusting the number of processing threads with respect to the input rate for synthetic data.}
  \label{fig:rateSeq}
\end{figure}

After showing \xxx{} results for this experiment, we want to further shed light on the benefits enabled by \xxx{}'s ultra-fast reconfigurations.
\autoref{fig:saso} shows a zoomed-in view of~\autoref{fig:rateSeq} at the time instants when the input rate increases at $290$ seconds.
We consider a transition duration starting at the moment the input rate changes until $\WS$ is exceeded (rectangular highlighted area in \autoref{fig:saso}).
Out of the transition duration, we show a baseline by extracting the desired number of processing threads for the specific input rate.
Within the transition duration, since \xxx{} can accommodate reconfigurations in negligible time, the desired number of threads can change gradually, following how the new rate affects the amount of computations until the window instances are filled with tuples coming with the new rate.

\begin{figure}[ht!]
    \centering
        \includegraphics[width=\linewidth]{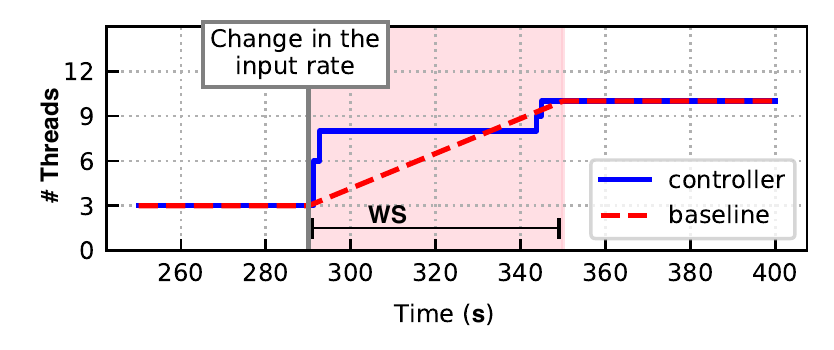}
    \caption{Zoomed-in view of the reconfigurations around second 290 from \autoref{fig:rateSeq}.}
    \label{fig:saso}
\end{figure}

\begin{figure*}
  \centering
  \includegraphics[width=\linewidth]{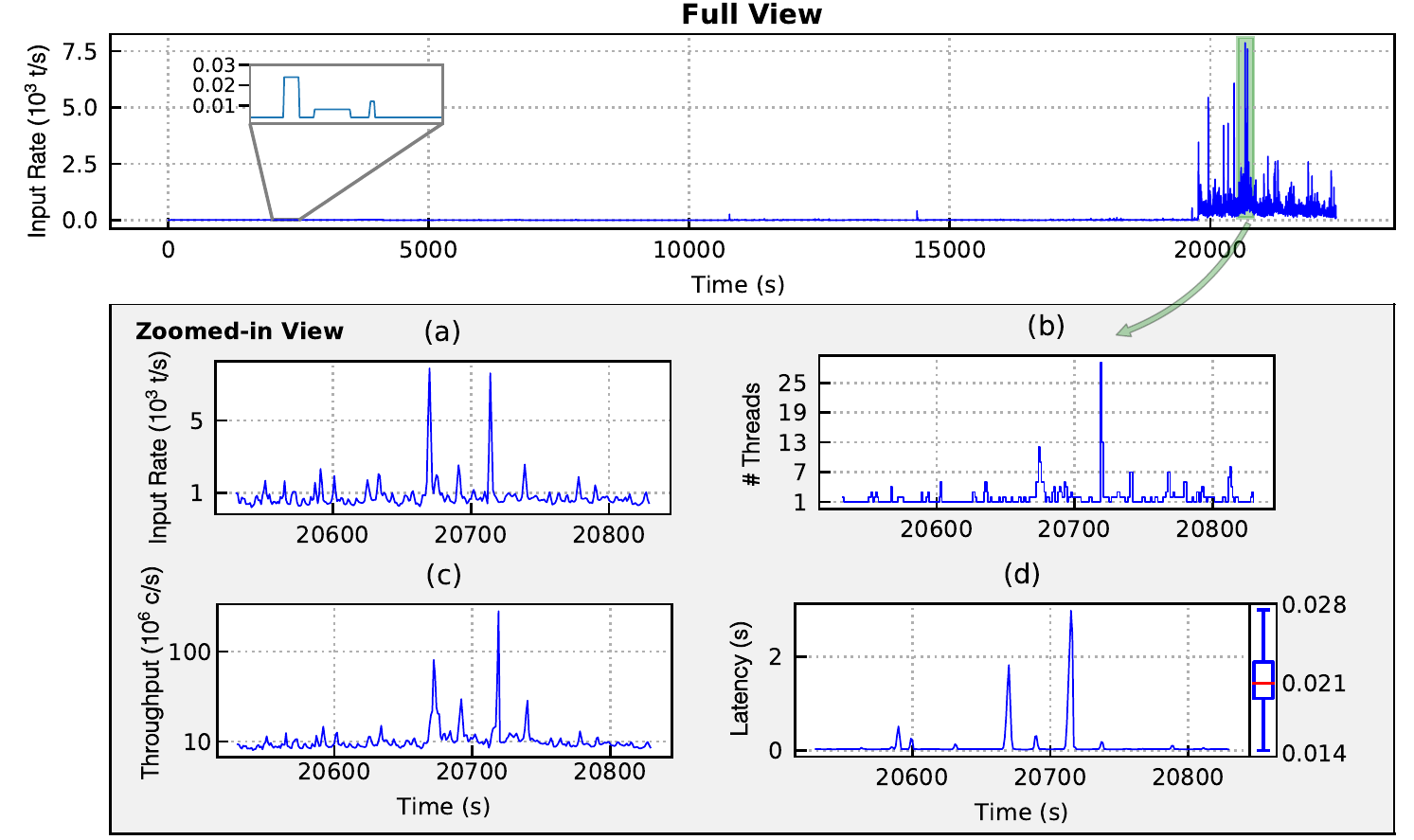}
  \caption{Results of adjusting the number of processing threads while processing the NYSE market data. The zoomed-in view indicates the details of the experiment during the time when the input rate is the highest.}
  \label{fig:financial}
\end{figure*}

\subsection{\xxx{} performance in a real-world setup ($Q_6$)}
\label{ssc:q6}

In this conclusive experiment, we evaluate \xxx{} using real-world data. We
use NYSE group reference data provided by the NYSE Market Data reporting authority (\url{ftp://ftp.nyxdata.com/}), containing six hours of trades registered in the NYSE and NASDAQ markets on July 30$^{th}$, 2018.
One major characteristic of this data is the abrupt and very frequent changes in the incoming rate.
This use-case examines correlations of the trades of the 
10 biggest companies of the day.
When looking into the characteristics of the data stream that these trades generate, 
we see its rate oscillates between $0$ and $8\,000$ t/s.
In this experiment, we run \joinplusopinmath{} as a self-Join, setting $\WS$ to 30 seconds and feeding twice a stream with schema $\langle \tau, \left[ \text{id}, \text{TradePrice}, \text{AveragePrice}\right] \rangle$,
where id, of type \texttt{string}, is the unique id of the company, TradePrice, of type \texttt{int}, is the price of the trade, and AveragePrice, of type \texttt{int}, is the average trade price of the corresponding company for the previous day.
The output schema $S_O$ is $\langle \tau, \left[ \text{l\_id}, \text{l\_price}, \text{r\_id}, \text{r\_price} \right] \rangle$, where l\_id, l\_price, r\_id, and r\_price are $t_L.\varphi[1]$, $t_L.\varphi[2]$, $t_R.\varphi[1]$ and $t_R.\varphi[2]$, respectively.

For the predicate function, we study the hedging (or negative correlation) of the stock prices among companies by computing the normalized distance of the stock price for each tuple with respect to its average price and then compare it with the other tuples. The normalized distance $ND_t$ of tuple $t$ is computed as 
$\frac{t.\varphi[2] - t.\varphi[3]}{t.\varphi[3]}$.
Consequently, the join results in an output tuple if:
$$t_R.\varphi[1] \neq t_L.\varphi[1] \bigwedge -1.05 \leq \frac{ND_{t_R}}{ND_{t_L}}$$

As shown in \autoref{fig:financial}, 
the hedge predicate can be run by a small number of threads most of the time, which also emphasizes the need for adjusting processing threads at runtime to utilize the resources.
Nonetheless, abrupt changes in the data rate require resource adjustments.
The zoomed-in view of \autoref{fig:financial} highlights the results for the time interval in which the highest peak of the input rate occurs.
\autoref{fig:financial}(a) indicates the input rate with the highest peak $\sim7\,600$ t/s.
The corresponding throughput and latency are also shown in \autoref{fig:financial}(c) and \autoref{fig:financial}(d), respectively. As shown in \autoref{fig:financial}(d), \xxx{} keeps the latency low (on average at 21 ms for the zoomed-in view, as shown by the box plot \autoref{fig:financial}(d), and 1 ms for the whole experiment) by frequently provisioning or decommissioning instances according to changes in the input rate. 

\separate
\section{Related Work}
\label{sec:rw}

\draft{Scalable and elastic stream processing are discussed in several related works (e.g.,~\cite{mayer2015predictable,mai2018chi,gulisano2012streamcloud}).
For a systematic review, we refer to \cite{gulisano2018enc}.
\xxx{} does not focus on a specific stateful operator~\cite{theodorakis2020lightsaber} nor on a specific policy to trigger reconfigurations~\cite{de2017proactive}.
Instead, it shows how shared-memory can be seamlessly leveraged to scale up stateful analysis while virtualizing the common APIs for SN parallelism (i.e., without altering how queries are usually composed and still being able to scale applications out via SN parallelism), something not previously discussed in the literature.}
\draft{In the following, we place particular focus on works about scale-up servers and intra-operator elasticity. Acknowledging that intra-operator elasticity is a way to boost performance in scale-up servers, we mostly discuss solutions that are not elastic when covering works dedicated to scale-up servers. For space reasons, we do not discuss elastic approaches not tailored to stream processing (e.g.,~\cite{de2016reconfiguration}).}

\draft{Focusing on scale-up servers, a first distinction can be made between works that aim at better leveraging the memory shared among CPU threads~\cite{zhang2019briskstream,mencagli2021windflow,theodorakis2020lightsaber} and works that, orthogonally to \xxx{}, focus on shared-memory in hybrid architectures, with a special focus on integrated~\cite{zhang2021fine,zhang2020finestream} and discrete~\cite{de2019gasser} CPU-GPU ones.}

\draft{For the former, differently from \xxx{}'s intra-operator memory sharing, \cite{zhang2019briskstream} focuses on the complementary aspect of inter-operator/intra-query memory sharing, to optimize operator placement for producer-consumer pairs e.g., based on NUMA distance. Complementary intra-query optimizations are also discussed in~\cite{mencagli2021windflow}, by formalizing basic stream processing tasks and defining composition/transformation rules that can be used to jointly boost performance metrics such as throughput and latency. An approach closer to \xxx{} is discussed in~\cite{theodorakis2020lightsaber}, with a special focus on streaming aggregation. Such an approach allows trading a higher degree of customization for the richer semantics and elasticity offered by \xxx{}.}

\draft{Focusing now on elasticity for intra-operator parallelization, we discuss related work based on goals and properties such as the number of threads that can change per reconfiguration, the roles of state transfer and the triggering mechanism, the support for correctness, and the reaction time vs overhead of frequent reconfigurations.}

Regarding changes in the number of threads, related works provide techniques for provisioning/decommissioning one thread at a time (e.g., \cite{schneider2009elastic,hochreiner2016elastic}) or more threads (e.g., \cite{heinze2014auto}), as in our case. Differently from \xxx{}, nonetheless, in \cite{heinze2014auto} the execution of operators being reconfigured is halted to ensure all required state transfers are completed before the processing of tuples resumes.

In connection to state transfers, a common challenge is that of deciding how to balance the load of a parallel operator, which is usually a combinatorial problem related to packing (cf. \cite{gulisano2012streamcloud,balkesen2013adaptive,heinze2013elastic,heinze2014auto} and references therein).
\draft{Some related work proposed techniques to make decisions that can reduce latency spikes~\cite{heinze2014latency}, recreating small states by replaying past tuples rather than serializing/deserializing states~\cite{gulisano2012streamcloud}, or by distributing the work to nodes through hashing, in ways that minimize the changes when rehashing~\cite{Gedik2014}. 
Since \xxx{} allows to re-balance work without any state transfer nor serialization/deserialization, it makes these issues orthogonal and existing techniques complementary.}

\draft{Regarding correctness, we note that such an aspect has been formalized by e.g., \cite{scalejoin,cederman2013concurrent} and is also referred to as \textit{determinism}~\cite{scalejoin,cederman2013concurrent}, \textit{safety}~\cite{Hirzel2014CSP} and \textit{semantic transparency}~\cite{gulisano2012streamcloud}. Here we show sufficient conditions for correct execution, also under reconfigurations.}

\draft{Focusing now on the issue of when to trigger reconfigurations, and the related trade-offs between overheads and time to react, related works can be distinguished into reactive or proactive approaches (cf.~\cite{gulisano2012streamcloud,martin2014scalable,hochreiner2016elastic,zacheilas2015elastic,Cardellini2016Elastic,de2017proactive,kumbhare2015reactive} and references therein). 
Earlier works proposed reactive strategies, e.g. with 
threshold-based rules based on CPU utilization~\cite{gulisano2012streamcloud, Fernandez2013Integrating, Heinze2014Scaling}, heuristic-based algorithms~\cite{Gedik2014Elastic}, and model-based approaches to enforce latency constraints~\cite{lohrmann2015elastic}.}
To make timely decisions, proactive approaches have gained more attention over the years. Model-based proactive methods such as~\cite{de2017proactive}, for instance, use a limited future time horizon to choose reconfigurations for timely execution, with two different resource usage characteristics to achieve high throughput and low latency.
Similar to~\cite{lohrmann2015elastic}, the technique proposed in~\cite{de2017proactive} uses queuing theory to model stream operators as G/G/1 queues.
More advanced queuing systems to model SPEs are discussed in~\cite{Cooper2019Queuing} to increase the accuracy of the results. However, the technique proposed in~\cite{Cooper2019Queuing} does not consider quality of service requirements and it is not optimized for low latency workload.
All in all, various complementary triggering mechanisms can be combined with \xxx{}.
As discussed by \cite{de2017proactive} with respect to SASO properties (i.e., Stability, Accuracy, Settling time, and Overshoot), elasticity builds on top of contrasting objectives. \xxx{} can imply better margins due to its ultra-fast reconfigurations.

\separate
\section{Conclusions and Future Work}
\label{sec:conc}

We introduced \xxx{}, a model as well as an implemented prototype for VSN parallel and elastic execution of stateful stream processing.
We started discussing how SN parallelism can potentially suffer data duplication and state transfer overheads, and how such overheads can be avoided when the instances of a parallel operator share memory.
For \xxx{} approach to be general while encapsulating the semantics of common stateful operators, we have introduced a generalized stateful operator \generalopplusinmath{}, and later discussed how VSN can be enabled for \generalopplusinmath{}'s instances by having the latter share their input tuples, output tuples, and state.
Together with an algorithmic description, we have provided a fully implemented prototype, which builds on the state-of-the-art stream-processing-tailored data structure ScaleGate.
Together with our theoretical contributions, we also provided an exhaustive evaluation, based both on synthetic and real data and comparing with various baselines.

To the best of our knowledge, this is the first work blending a formal approach with correctness proofs about parallel/elastic stateful stream processing with a model and prototype complementary to that of SN-based solutions. \xxx{} can pave the road for hybrid approaches that extend from intra- to inter-node parallel and elastic solutions.

\section*{Acknowledgments}
Work supported by the Swedish Foundation for Strategic Research (SSF) -- project ``FiC'' grant GMT14-0032 -- and the Swedish Research Council (Vetenskapsr\aa{}det) -- projects ``HARE'' grant 2016-03800, and ``PSI'' grant 2020-05094.

\separate
\bibliographystyle{IEEEtran}
\bibliography{refs}

\begin{IEEEbiography}[{\vskip -3\baselineskip\includegraphics[width=1in,height=1in,clip,keepaspectratio]{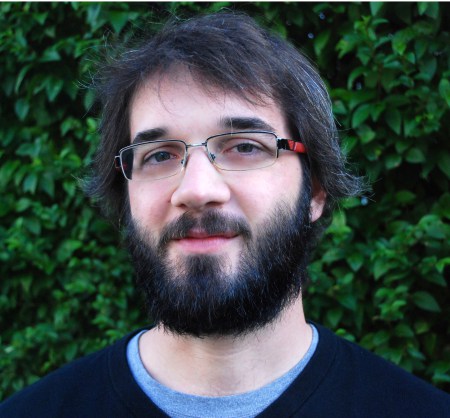}}]%
{Vincenzo Gulisano} is an Associate Professor at Chalmers University of Technology. His research focuses on data processing and distributed / parallel / elastic and fault-tolerant data streaming. Dr. Vincenzo Gulisano holds a PhD in Computer Science from the Polytechnic University of Madrid, Spain.
\end{IEEEbiography}
\vskip -2\baselineskip plus -1fil
\begin{IEEEbiography}[{\vskip -3\baselineskip\includegraphics[width=1in,height=1in,clip,keepaspectratio]{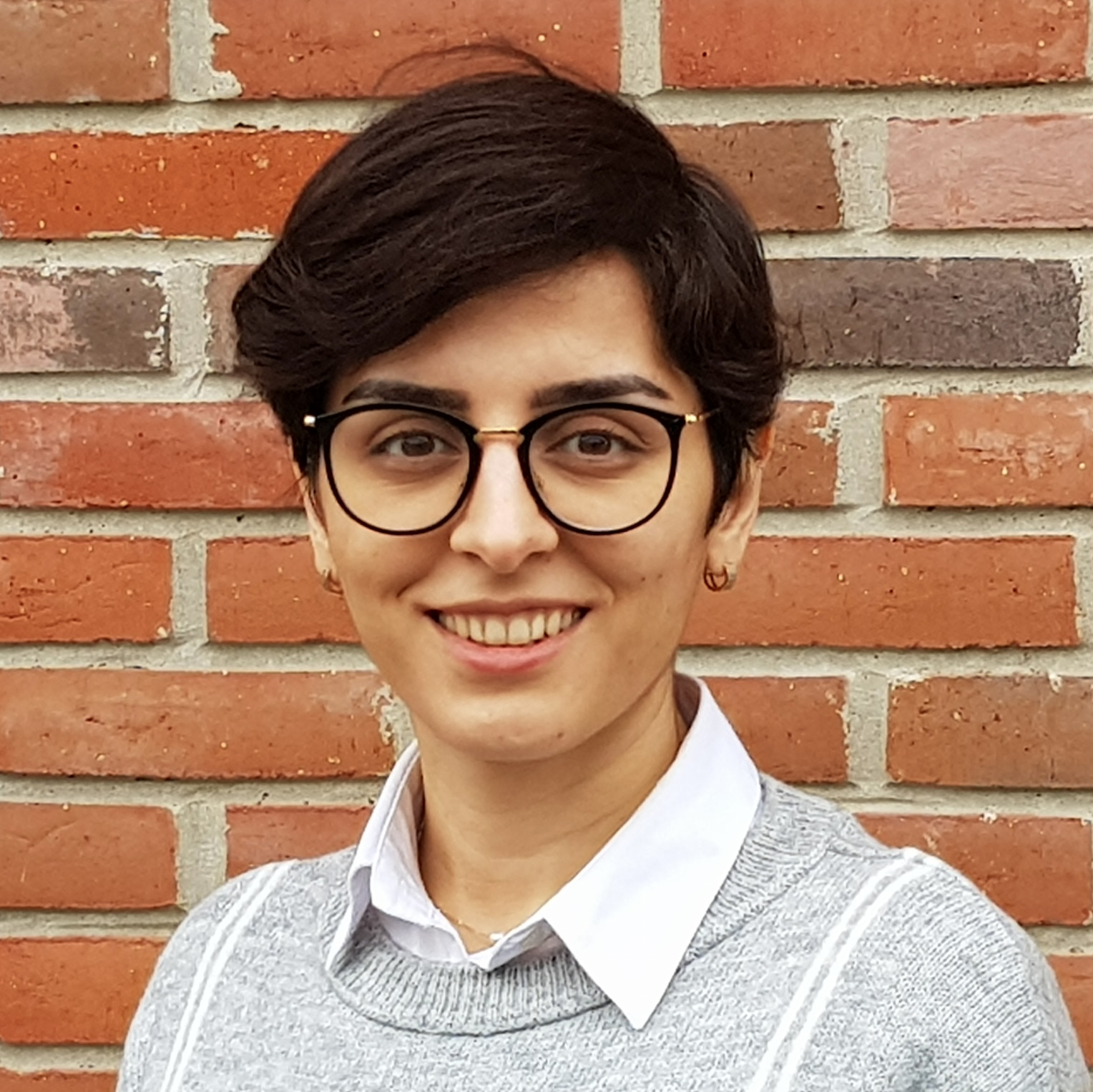}}]%
{Hannaneh Najdataei} received her BSc degree in Software Engineering and her MSc degree in Artificial Intelligence from Shiraz University in Iran. She is currently a PhD student at Chalmers University of Technology. Her research focuses on parallel programming, data stream processing and cyber–physical systems.
\end{IEEEbiography}
\vskip -2\baselineskip plus -1fil
\begin{IEEEbiography}[{\vskip -3\baselineskip\includegraphics[width=1in,height=1in,clip,keepaspectratio]{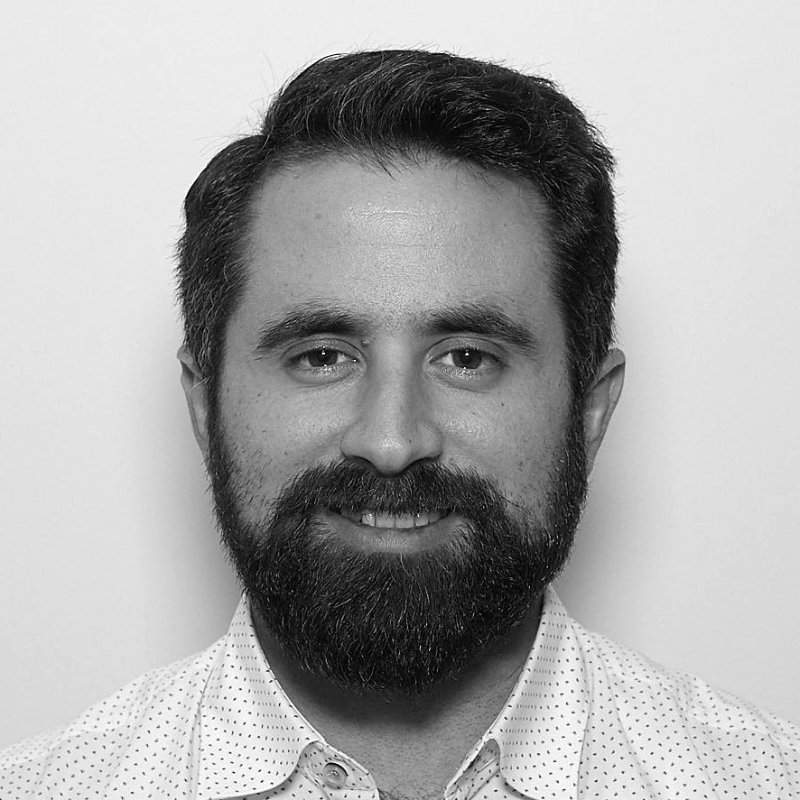}}]%
{Yiannis Nikolakopoulos} received the
B.Sc. in Mathematics from the National
and Kapodistrian University of
Athens, Greece and the Ph.D. degree
in Computer Science and Engineering
from Chalmers University of Technology,
Gothenburg, Sweden. Currently, he is a system software engineer at ZeroPoint Technologies AB. His research interests include memory management, memory compression, concurrent data structures and stream processing on multicore and many-core systems.
\end{IEEEbiography}
\vskip -2\baselineskip plus -1fil
\begin{IEEEbiography}[{\vskip -3\baselineskip\includegraphics[width=1in,height=1in,clip,keepaspectratio]{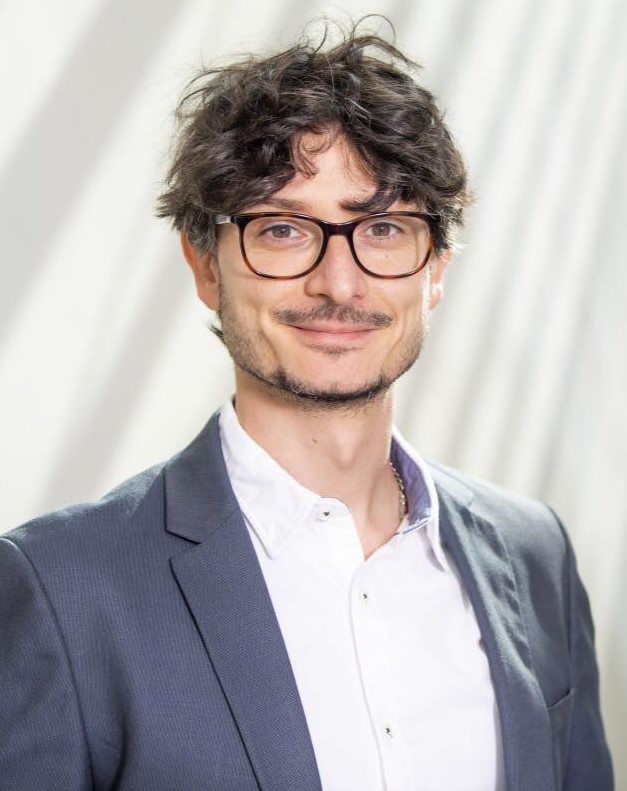}}]{Alessandro~V. Papadopoulos}
(Senior Member, IEEE) received the B.Sc. and M.Sc.
(summa cum laude) degrees in computer engineering
and the Ph.D. degree (Hons.) in information technology from the Politecnico di Milano, Italy, in 2008, 2010, and 2013, respectively.
He is an Associate Professor of Computer Science
with M{\"a}lardalen University, V{\"a}ster{\aa}s, Sweden. His
research interests include robotics, control theory,
real-time systems, autonomic computing.
\end{IEEEbiography}
\vskip -2\baselineskip plus -1fil
\begin{IEEEbiography}[{\includegraphics[width=1in,height=1in,clip,keepaspectratio]{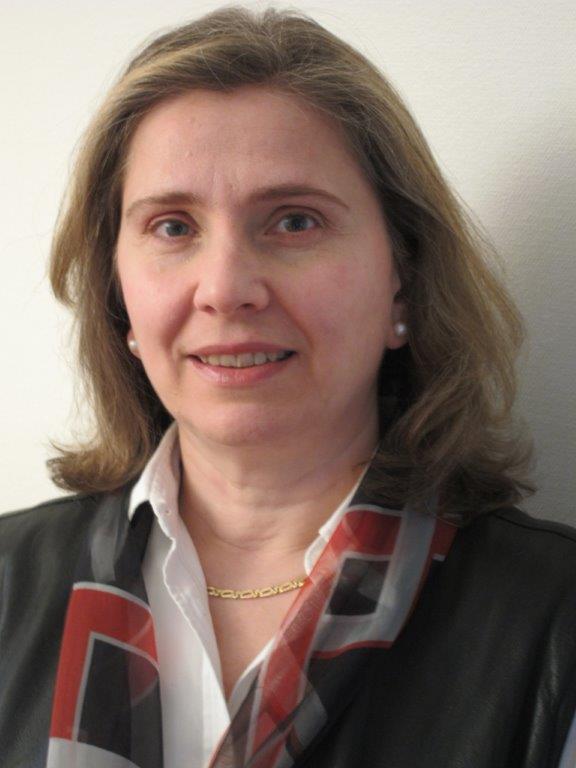}}]%
{Marina Papatriantafilou} is
Associate Professor at Chalmers University of Technology. Earlier, she was with the Max-Planck Institute for Computer Science, Saarbruecken and CWI, Amsterdam. She received her PhD degree from the Computer Science and Informatics Dept., Patras University and is a member of Network of National Contacts ACM-WE NeNaC.
Her research interests include:  efficient parallel, distributed stream processing and applications in multiprocessor, multicore and distributed, cyber-physical systems; synchronization, consistency, fault-tolerance.
\end{IEEEbiography}
\vskip -2\baselineskip plus -1fil
\begin{IEEEbiography}[{\includegraphics[width=1in,height=1in,clip,keepaspectratio]{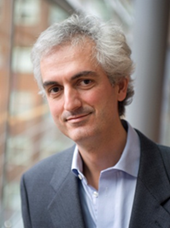}}]%
{\bf Philippas Tsigas} received the B.Sc. degree in mathematics and the Ph.D. degree in computer engineering and informatics from the University of Patras, Greece. He was at the National Research Institute for Mathematics and Computer Science, Amsterdam, The Netherlands (CWI), and at the Max-Planck Institute for Computer Science, Saarbrucken, Germany, before. At present, he is a professor at Chalmers University of Technology, Sweden. His research interests include concurrent data structures and algorithmic libraries for multiprocessor and many-core systems, communication and synchronization in parallel systems, power aware computing, fault-tolerant computing, autonomic computing, scalable data streaming.
\end{IEEEbiography}

\clearpage
\appendices
\section{Table of symbols (in alphabetical order)}
\label{apx:symbols}

{\small
\begin{tabular}{@{}lp{11cm}@{}}
\toprule
\textbf{Symbol} & \textbf{Description} \\
\midrule
$\langle \tau,\dots, \left[ \varphi[1], \varphi[2], \dots \right] \rangle$ & Combined notation for a tuple \\ 
$A(\WA,\WS,1,\keybysingle,\psipar,S,f_A,f_R)$ & Combined notation for Aggregate operator \\ 
$\aggplusop(\WA,\WS,1,\keybymulti,\psipar,S,f_A,f_R)$ & Combined notation for Aggregate operator that relies on $\keybymulti$ instead of $\keybysingle$ \\ 
$\gamma$ & Event time shared by \generalopplusinmath{} instances at which a reconfiguration takes place \\ 
$\delta$ & Smallest increment of event time of an SPE  \\ 
$D$ & Generic downstream operator/egress of $O$/$\generalopplus$  \\ 
DAG & Directed Acyclic Graph \\ 
$e$ & epoch during the execution of a certain operator \\ 
$f_A$ & Aggregate function of \aggop{} \\ 
$f_J$ & Join function of \joinop{} \\ 
\opmappinginmath{} & Function mapping keys to operators instances \\ 
$\keybymulti$ & Multi Key-by function that maps each tuple $t$ to an arbitrary number of key values \\ 
$f_O$ & Function used by \generalopplusinmath{} to produce the result from a window instance \\ 
$f_R$ & Reduce function of \aggop{} \\ 
$f_S$ & Function used by \generalopplusinmath{} to shift a window instance \\ 
$\keybysingle$ & Single Key-by function that maps each tuple $t$ to exactly one key value \\ 
$f_U$ & Function used by \generalopplusinmath{} to update a window instance upon reception of $t$ \\ 
\texttt{forward}/\texttt{forwardSN}/\texttt{forwardVSN} & Methods used by an operator/ingress instance to route tuples downstream (for $O$, \generalopplusinmath{} in SN, and \generalopplusinmath{} in VSN setups, respectively) \\ 
$J(\WA,\WS,2,\keybysingle,\psipar,S,f_J)$ & Combined notation for Join operator \\ 
$\joinplusop(\WA,\WS,2,\keybymulti,\psipar,S,f_J)$ & Combined notation for Aggregate operator that relies on $\keybymulti$ instead of $\keybysingle$ \\ 
\mapop{} & Map/FlatMap operator \\ 
$O(\WA,\WS,I,\keybysingle,\psipar,S,f_1,f_2,\ldots)$ & Combined notation for a generic stateful operator. Parameters $\WA$, $\WS$, $I$, $\keybysingle$, $\psipar$, and $S$ are covered in the table. Parameters $f_1,f_2$ refer to operator-specific functions to update $O$'s state and produce output tuples \\ 
$o_j$ & $j$-th instance of $O$ \\ 
$\generalopplus(\WA,\WS,I,\keybymulti,\opmapping,\psipar,S,f_U,f_O,f_S)$ & Combined notation for \xxx{}'s generalized operator \\ 
$\generalopplusinst_j$ & $j$-th instance of \generalopplusinmath{} \\ 
$\mathbb{O}$ & Set of parallel instances running for a certain during epoch $e$\\ 
$\Pi(X)$ & Parallelism degree of ingress/operator/egress $X$  \\ 
\texttt{process}/\texttt{processSN}/\texttt{processVSN} & Methods used by an operator instance to process each new input tuple (for $O$, \generalopplusinmath{} in SN, and \generalopplusinmath{} in VSN setups, respectively) \\ 
$q_{a,b}$ & Queue connecting operator instances (or ingress/egress) $a$ and $b$\\ 
$S$ & Schema of the tuples of the referenced stream  \\ 
$\opstate_j$ & State of $o_j/\generalopplusinst_j$ (when locally maintained at $o_j/\generalopplusinst_j$ \\ 
$\opstate$ & State of $\generalopplusinst_j$ (when shared by all $\generalopplus$ instances \\ 
SN & Shared Nothing \\ 
SPE & Stream Processing Engine \\ 
$t.\tau$, $t.\varphi$ & Timestamp and payload ($\ell$-th sub-attribute is $t.\varphi[\ell]$) of tuple $t$ \\ 
$U_i$ & The $i$-th generic upstream operator/ingress of $O$/$\generalopplus$  \\ 
$u_{i,j}$ & The $j$-th instance of $U_i$ \\ 
VSN & Virtual Shared Nothing \\ 
$w=\langle \zeta,l,k \rangle$ & Combined notation for a window instance $w$: state $\zeta$, left inclusive boundary $l$, and key $k$.\\ 
\watermarkofinmath{o_j} & Watermark of the $j$-th instance of operator $O$ at clock wall time $\omega$ \\ 
$\WA$ & Window advance \\ 
$\WS$ & Window size \\ 
$\psipar$ & Window Type maintained by a stateful operator (single or multi)  \\ 
$\X/\X_{in}/\X_{out}$ (Tuple Buffer) & Data object \xxx{} uses to connect $U_i$ with \generalopplusinmath{} and \generalopplusinmath{} with $D$ (cf.~\autoref{tab:xapi})\\ 
\bottomrule
\end{tabular}}
\clearpage
\section{Watermark example}
\label{apx:exampleofwatermark}

\begin{figure}[ht!]
  \centering
  \includegraphics[width=\linewidth]{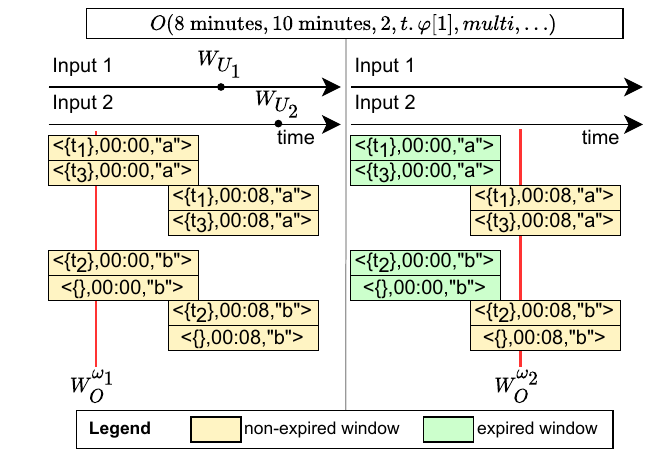}
  \caption{\draft{Internal state of operator $O$ before (left side of the figure) and after (right side of the figure) updating its watermark from $W_{O}^{\omega_1}$ to $W_{O}^{\omega_2}$.}}
  \label{fig:watermarkexample}
\end{figure}

\draft{
\autoref{fig:watermarkexample} shows how the internal state of Operator $O$ from~\autoref{fig:windows_examples} changes from wall-clock time $\omega_1$ to $\omega_2$, upon the processing of two tuples, one from each upstream peer, that carry new watermark values. As shown, none of the window instances maintained by $O$ before updating its watermark (left side of~\autoref{fig:watermarkexample}) is expired, since incoming tuples, with a timestamp greater than $W_{O}^{\omega_1}$, could still contribute to any of such window instances. Once $O$'s watermark changes to $W_{O}^{\omega_2}$, no new input tuple can contribute to window instances starting at 00:00. Hence, such windows are expired and their corresponding output tuple can be created and forwarded to $O$'s downstream peer(s).
}
\section{Example for Corollary~\ref{cor:duplicationviamap}}
\label{apx:appendixcorollary}

\begin{figure}[ht!]
  \centering
  \includegraphics[width=\linewidth]{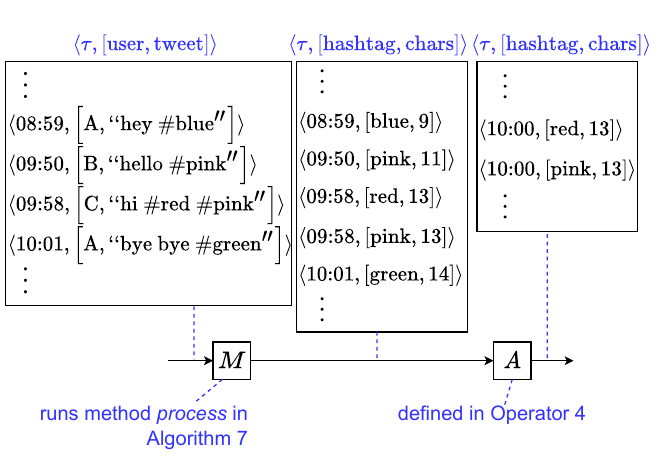}
  \caption{\draft{Sample execution of a $\mapop{}$ and an $\aggop{}$ that, together, implement the semantics of the running example from~\autoref{sec:introduction}}}
  \label{fig:corollary1example}
\end{figure}

\renewcommand*{\algorithmcfname}{Algorithm}
\begin{algorithm}[h]
\setcounter{algocf}{6}
\footnotesize
\SetAlgoLined
\DontPrintSemicolon
\SetKwInput{KwParameters}{Instance-local variables}
\SetKwInput{KwAuxiliaryFunctions}{Auxiliary functions}
\SetKwProg{Proc}{Function}{}{}
\SetKwProg{Class}{Class}{}{}
\BlankLine
\Proc{\FuncSty{process}($t$)}
{
    \lFor{$h \in $\FuncSty{ hashtags($t.\varphi[2]$)}}{
        \FuncSty{forward}($\langle t.\tau, \left[ h, \FuncSty{length}(t.\varphi[2]) \right] \rangle$)\label{ex:corollaryduplicateviamap_map:output} 
    }
}
\caption{\draft{Method \texttt{process} for the $\mapop{}$ used in~\autoref{fig:corollary1example} to compute the longest tweet  on a per-hashtag basis.}}
\label{ex:corollaryduplicateviamap_map}
\end{algorithm}

\renewcommand*{\algorithmcfname}{Operator}

\begin{algorithm}[h]
\setcounter{algocf}{0}
\footnotesize
\SetAlgoLined
\DontPrintSemicolon
\SetKwInput{KwParameters}{Instance-local variables}
\SetKwInput{KwAuxiliaryFunctions}{Auxiliary functions}
\SetKwProg{Proc}{Function}{}{}
\SetKwProg{Class}{Class}{}{}
\BlankLine
\nonl $\aggop(30m,60m,1,\keybysingle,\text{multi},\opmapping,S_O,f_A,f_R)$, where:\;  
\BlankLine
\Proc{$\keybysingle(t)$}
{
    \Return $t.\varphi[1]$\;
}
\Proc{$\opmappingfun(k)$}
{
    \Return \FuncSty{hash($k$)}\%$\Pi(\aggop)$ \;
}
\Class{$\zeta$}
{
    \FuncSty{long count}\;
}
\Proc(\tcp*[h]{Update count based on $t.\varphi[2]$}){$f_A(w,t)$}
{
    \If{$t.\varphi[2]$$>w.\zeta.$\FuncSty{count}}{
        $w.\zeta.$\FuncSty{count}$\xleftarrow{}$$t.\varphi[2]$\;
    }
    \Return \{$w.\zeta$\}\;
}
\Proc{$f_R(w,t)$}
{
    \Return $\{w.k,w.\zeta.\text{count}\}$\;
}
\caption{\draft{User-defined parameters for the $\aggop{}$ used in~\autoref{fig:corollary1example} to compute the longest tweet on a per-hashtag basis.}}
\label{ex:corollaryduplicateviamap_agg}
\end{algorithm}

\draft{
\autoref{fig:corollary1example} illustrates how a $\mapop{}$ and an $\aggop{}$ operator can be jointly used to implement the semantics of an $\aggplusop{}$, according to Corollary~\ref{cor:duplicationviamap}.
For each stream, the figure shows an excerpt of the tuples observed as input/output of the two operators.
$\mapop{}$ implements the \texttt{process} function presented in~\autoref{ex:corollaryduplicateviamap_map}, creating and forwarding one output tuple for each hashtag found in an input tuple (\autoref{ex:corollaryduplicateviamap_map} L\ref{ex:corollaryduplicateviamap_map:output}). For instance, forwarding tuples $\langle \text{09:58}, \left[\text{red}
, 13\right] \rangle$ and $\langle \text{09:58}, \left[\text{pink}
, 13\right] \rangle$ upon processing of $\langle \text{09:58}, \left[\text{C}, \text{``hi \#red \#pink''}\right] \rangle$.
}

\draft{
$\aggop{}$ is defined according to Operator~\ref{ex:corollaryduplicateviamap_agg}.
$\aggop{}$ aggregates tuples over a sliding window with $\WA{}$ and $\WS{}$ of 30 and 60 minutes, respectively.
According to $\keybysingle$, tuples sharing the same hashtag are aggregated together, and routed to the same instance of $\aggop{}$ by $\opmapping$ (cf.~\autoref{ssc:snparallelism}). In the example, output tuples $\langle \text{10:00}, \left[\text{red}, 13\right] \rangle$ and $\langle \text{10:00}, \left[\text{pink}, 13\right] \rangle$ are the results of aggregating the input tuples covering the event time $[\text{09:00},\text{10:00})$, $\langle \text{09:50}, \left[\text{pink}, 11\right] \rangle$ and $\langle \text{09:58}, \left[\text{pink}, 13\right] \rangle$ for key \textit{pink}, and $\langle \text{09:58}, \left[\text{red}, 13\right] \rangle$ for key \textit{red}.
}

\section{Example Operators}
\label{apx:examples}

\begin{algorithm}[h]
\footnotesize
\SetAlgoLined
\DontPrintSemicolon
\SetKwInput{KwParameters}{Instance-local variables}
\SetKwInput{KwAuxiliaryFunctions}{Auxiliary functions}
\SetKwProg{Proc}{Function}{}{}
\SetKwProg{Class}{Class}{}{}
\BlankLine
\nonl $\aggplusop(\WA,\WS,1,\keybymulti,\text{multi},S_O,\opmapping,f_U,f_O,-)$, where:\;  
\BlankLine
\Proc{$\keybymulti(t)$}
{
    $\mathbb{K}\xleftarrow{}\{\}$\; 
    \lFor{$h \in $\FuncSty{ hashtags($t.\varphi[2]$)}}{
        $\mathbb{K} \xleftarrow{} \mathbb{K} \cup h$ \label{ex:a:hash} 
    }
    \Return $\mathbb{K}$\;
}
\Proc{$\opmappingfun(k)$}
{
    \Return \FuncSty{hash($k$)}\%$\Pi(\aggplusop)$ \;
}
\Class{$\zeta$}
{
    \FuncSty{long count}\;
}
\Proc(\tcp*[h]{Update count based on $t.\varphi[2]$}){$f_U(w,t)$}
{
    \If{\FuncSty{length($t.\varphi[2]$)}$>w.\zeta.$\FuncSty{count}}{
        $w.\zeta.$\FuncSty{count}$\xleftarrow{}$\FuncSty{length($t.\varphi[2]$)}\;
    }
    \Return \{$w.\zeta$\}\;
}
\Proc{$f_O(w,t)$}
{
    \Return $\{w.k,w.\zeta.\text{count}\}$\;
}
\caption{User-defined parameters (cf.~\autoref{sec:generalizedmodel}) for an \aggplusopinmath{} computing the longest tweet  on a per-hashtag basis.}
\label{ex:a}
\end{algorithm}

\draft{
This appendix provides several complete examples.
It begins by providing one example for the \aggplusopinmath{} and one example for the \joinplusopinmath{} operators introduced in~\autoref{sec:generalizedmodel}.
The latter operator is one of the operators used to evaluate \xxx{} in~\autoref{sec:eval}.
The remaining examples are for other operators that are also used in~\autoref{sec:eval}.
}
\subsubsection*{\aggplusopinmath{} -- compute longest tweet per hashtag}
This \aggplusopinmath{} operator is the running example from~\autoref{sec:introduction}.
We assume that the input schema is defined as $\langle \tau, \left[ \text{user}, \text{tweet} \right] \rangle$, while the output schema $S_{O}$ as $\langle \tau, \left[ \text{hashtag}, \text{chars} \right] \rangle$, where \textit{hashtag} and  \textit{chars} represent the number of characters \textit{chars} of the longest tweet observed for each \textit{hashtag}.
As shown in Operator~\ref{ex:a}, \keybymultiinmath{} returns one key for each hashtag contained in the $tweet$ attribute of $t$ (L\ref{ex:a:hash}). Each hashtag is then matched with one instance by \opmappingfuninmath{}, hashing the key and returning its value modulo the parallelism degree of \aggplusopinmath{}.
The state associated with each window instance $w$ is a counter. Function $f_U$ is invoked to update such count, while $f_O$ is invoked to forward the hashtag ($k$) and count of each expired window instance. Since $\psipar = \text{multi}$, each expired window instance is simply removed after being passed as parameter to $f_O$.

\begin{algorithm}[h]
\footnotesize
\SetAlgoLined
\DontPrintSemicolon
\SetKwInput{KwParameters}{Instance-local variables}
\SetKwInput{KwAuxiliaryFunctions}{Auxiliary functions}
\SetKwProg{Proc}{Function}{}{}
\SetKwProg{Class}{Class}{}{}

\BlankLine
\nonl $\joinplusop(\WA,\WS,2,\keybymulti,\text{single},S_O,\opmapping,f_U,-,-)$, where:\;  
\BlankLine

\Proc{$\keybymulti(t)$}
{
    \Return $\{1,\ldots,1000\}$\;\label{alg:j:keyby} 
}
\Proc{$\opmappingfun(k)$}
{
    \Return \FuncSty{hash($k$})\%$\Pi(\joinplusop)$\; \label{alg:j:fmu}
}
\Class{$\zeta$}
{
    \FuncSty{long $c$} \tcp{tuple counter}\label{alg:j:zetastart}
    \FuncSty{queue<$t$> $T$} \tcp{previous tuples}\label{alg:j:zetastop}
}
\Proc{$f_U(\{w_1,w_2\},t)$}
{
    $\mathbb{T}_O\xleftarrow{}\{\}$ \tcp{create set for output tuples}
    $w_1.\zeta.c\xleftarrow{}w_1.\zeta.c+1$ \tcp{increase $c$ of win inst.}
    $w_2.\zeta.c\xleftarrow{}w_2.\zeta.c+1$\;
    \If(\tcp*[h]{set this/opposite win. inst.}){$t$ from $U_1$}{\label{alg:j:joinstart}
        \FuncSty{this}$_w\xleftarrow{}w_1$\;
        \FuncSty{opp}$_w\xleftarrow{}w_2$\;
    } \Else {
        \FuncSty{this}$_w\xleftarrow{}w_2$\;
        \FuncSty{opp}$_w\xleftarrow{}w_1$\;
    }
    \While(\tcp*[h]{purge}){\FuncSty{opp}$_w.\zeta.T[0].\tau+\WS<t.\tau$}{
        \FuncSty{opp}$_w.\zeta.T$.\FuncSty{dequeue()}
    }
    \For(\tcp*[h]{match}){$t' \in $\FuncSty{ opp}$_w.\zeta.T$}{
        \lIf{$t'$ and $t$ match}{
        $\mathbb{T}_O \xleftarrow{} \mathbb{T}_O \cup \{t'.\varphi,t.\varphi\}$
        }
    }
    \If(\tcp*[h]{store $t$}){$c\%1000=$\FuncSty{this}$_w.k$}{
        \FuncSty{this}$_w.\zeta.T$.\FuncSty{enqueue($t$)}\;
    }
    \Return $\mathbb{T}_O$ \tcp{return results}\label{alg:j:joinstop}
}
\caption{User-defined parameters (cf.~\autoref{sec:generalizedmodel}) for a \joinplusopinmath{} implementing ScaleJoin~\cite{scalejoin}.}
\label{ex:j}
\end{algorithm}

\renewcommand*{\algorithmcfname}{Algorithm}

\begin{algorithm}[h]
\setcounter{algocf}{7}
\footnotesize
\SetAlgoLined
\DontPrintSemicolon
\SetKwInput{KwParameters}{Instance-local variables}
\SetKwInput{KwAuxiliaryFunctions}{Auxiliary functions}
\SetKwProg{Proc}{Function}{}{}
\SetKwProg{Class}{Class}{}{}
\BlankLine
\Proc{\FuncSty{process}($t$)}
{
    \lFor{$w \in $\FuncSty{split($t.\varphi[2]$)}}{
        \FuncSty{forward}($\langle t.\tau, \left[w\right] \rangle$)
    }
}
\caption{\draft{Method \texttt{process} for the $\mapop{}$ used in~\autoref{ssc:q1} in Flink's \texttt{wordcount} implementation.}}
\label{ex:flink:wordcout:map}
\end{algorithm}

\renewcommand*{\algorithmcfname}{Operator}

\begin{algorithm}[h]
\setcounter{algocf}{3}
\footnotesize
\SetAlgoLined
\DontPrintSemicolon
\SetKwInput{KwParameters}{Instance-local variables}
\SetKwInput{KwAuxiliaryFunctions}{Auxiliary functions}
\SetKwProg{Proc}{Function}{}{}
\SetKwProg{Class}{Class}{}{}
\BlankLine
\nonl $\aggop(60s,120s,1,\keybysingle,\text{multi},\opmapping,S_O,f_A,f_R)$, where:\;  
\BlankLine
\Proc(\tcp*[h]{implementation for \FuncSty{wordcount}}){$\keybysingle(t)$}
{
    \Return $t.\varphi[1]$\;
}\Proc(\tcp*[h]{implementation for \FuncSty{paircount}}){$\keybysingle(t)$}
{
    \Return $\langle t.\varphi[1],t.\varphi[2] \rangle$\;
}
\Proc{$\opmappingfun(k)$}
{
    \Return \FuncSty{hash($k$)}\%$\Pi(\aggop)$ \;
}
\Class{$\zeta$}
{
    \FuncSty{long count}\;
}
\Proc{$f_A(w,t)$}
{
    $w.\zeta.$\FuncSty{count}$\xleftarrow{}w.\zeta.$\FuncSty{count}$+1$\;
    \Return \{$w.\zeta$\}\;
}
\Proc{$f_R(w,t)$}
{
    \Return $\{w.k,w.\zeta.\text{count}\}$\;
}
\caption{\draft{User-defined parameters for the $\aggop{}$ used in~\autoref{ssc:q1} in Flink's \texttt{wordcount}/\texttt{paircount} implementation.}}
\label{ex:flink:paircount:aggregate}
\end{algorithm}

\renewcommand*{\algorithmcfname}{Algorithm}

\begin{algorithm}[h]
\setcounter{algocf}{8}
\footnotesize
\SetAlgoLined
\DontPrintSemicolon
\SetKwInput{KwParameters}{Instance-local variables}
\SetKwInput{KwAuxiliaryFunctions}{Auxiliary functions}
\SetKwProg{Proc}{Function}{}{}
\SetKwProg{Class}{Class}{}{}
\KwParameters{}
$B$ \tcp{Maximum dist. for 2 words to form a pair}
\BlankLine
\Proc{\FuncSty{process}($t$)}
{
    $\mathbb{W}\xleftarrow{}$\FuncSty{split($t.\varphi[2]$)} \;
    \For{$i \in 0,\ldots,|\mathbb{W}|-1$}{
    \For{$j \in i+1,\ldots,|\mathbb{W}|-1$}{
    \lIf{$j-i\leq B$}{
        \FuncSty{forward}($\langle t.\tau, \left[\mathbb{W}[i],\mathbb{W}[j]\right] \rangle$)
    }
    }
    }
}
\caption{\draft{Method \texttt{process} for the $\mapop{}$ used in~\autoref{ssc:q1} in Flink's \texttt{paircount} implementation.}}
\label{ex:flink:paircount:map}
\end{algorithm}

\subsubsection*{\joinplusopinmath{} -- ScaleJoin}

To provide an example of a \joinplusopinmath{}, we now show how \generalopplusinmath{} can be used to implement ScaleJoin~\cite{scalejoin}, which performs a Cartesian join of all tuples belonging to two windows. 

To run in a parallel and skew-resilient fashion, it delivers each input tuple (from any of the two input streams) to all instances, and has each instance compare it with a share of previous tuples (from any of the two input streams). Each tuple is stored, in a round-robin fashion, by exactly one of the instances. 
We assume in this case that $S_O$ is the concatenation of schemas from the two input streams.

As shown in Operator~\ref{ex:j}, this implementation maintains a fixed number of keys, larger than $\Pi(\joinplusop)$ (1000 in the example). All keys are returned for each tuple by \keybymultiinmath{} (L\ref{alg:j:keyby}). Hence, each instance will be given the chance of running $f_U$ for their share of keys (L\ref{alg:j:fmu}). The state associated with each window instance consists of counter $c$ and a \texttt{queue} of tuples (L\ref{alg:j:zetastart}-\ref{alg:j:zetastop}).
Since $I=2$, function $f_U$ is invoked passing two window instances together with $t$.
The counter of each pair of window instances (for all keys) is consistently increased by one for each window instance upon reception of the same tuple for all $j^+_i$, because \elasticscalegate{}$_{out}$ delivers all tuples in the same order to all instances. Afterward, $t$ is used to purge all stale tuples from the window instance opposite to $t$'s stream and subsequently matched with all remaining tuples in such window instance.
Finally, $t$ is stored, in a round-robin fashion based on $c$, in the window instance associated with exactly one key by one instance (L\ref{alg:j:joinstart}-\ref{alg:j:joinstop}).

\subsubsection*{\draft{Other evaluation operators}}

\draft{This section covers additional examples for other operators used in~\autoref{sec:eval}.
Algorithm~\ref{ex:flink:wordcout:map} and Operator~\ref{ex:flink:paircount:aggregate} refer to the \texttt{process} method of $\mapop{}$ and the implementation of $\aggop{}$ in Flink, respectively, for the \texttt{wordcount} experiment in \autoref{ssc:q2}, while Algorithm~\ref{ex:flink:paircount:map} and Operator~\ref{ex:flink:paircount:aggregate} are the ones used in Flink for $\mapop{}$ and $\aggop{}$, respectively, for the \texttt{paircount} experiment in \autoref{ssc:q2}.
In this case, we re-use the same Operator and specialize function $\keybysingle$ for compact notation.
Operator~\ref{ex:stretch:wordpaircount} is the $\aggplusop{}$ used in \xxx{} for the \texttt{wordcount} and \texttt{paircount} experiments in \autoref{ssc:q2}.
Also in this case, we re-use the same Operator and specialize $\keybymulti$ for compact notation.
Finally, Operator~\ref{ex:o2} is used to measure the maximum throughput/minimum latency when the performance bottleneck is given by data sharing/sorting (\autoref{ssc:q2}).}

\renewcommand*{\algorithmcfname}{Operator}

\begin{algorithm}[h]
\setcounter{algocf}{4}
\footnotesize
\SetAlgoLined
\DontPrintSemicolon
\SetKwInput{KwParameters}{Instance-local variables}
\SetKwInput{KwAuxiliaryFunctions}{Auxiliary functions}
\SetKwProg{Proc}{Function}{}{}
\SetKwProg{Class}{Class}{}{}
\BlankLine
\nonl $\aggplusop(\WA,\WS,1,\keybymulti,\text{multi},S_O,\opmapping,f_U,f_O,-)$, where:\;  
\BlankLine
\Proc(\tcp*[h]{implementation for \FuncSty{wordcount}}){$\keybymulti(t)$}
{
    $\mathbb{K}\xleftarrow{}\{\}$\;
    \lFor{$w \in $\FuncSty{split($t.\varphi[2]$)}}{
        $\mathbb{K} \xleftarrow{} \mathbb{K} \cup w$
    }
    \Return $\mathbb{K}$\;
}
\Proc(\tcp*[h]{implementation for \FuncSty{paircount}. Parameter $B$ (the max. dist. for 2 words to form a pair) is a local variable of $\keybymulti$.}){$\keybymulti(t)$}
{
    $\mathbb{K}\xleftarrow{}\{\}$\;
    $\mathbb{W}\xleftarrow{}$\FuncSty{split($t.\varphi[2]$)} \;
    \For{$i \in 0,\ldots,|\mathbb{W}|-1$}{
    \For{$j \in i+1,\ldots,|\mathbb{W}|-1$}{
    \lIf{$j-i\leq B$}{
         $\mathbb{K} \xleftarrow{} \mathbb{K} \cup \langle \mathbb{W}[i],\mathbb{W}[j] \rangle$
    }
    }
    }
    \Return $\mathbb{K}$\;
}
\Proc{$\opmappingfun(k)$}
{
    \Return \FuncSty{hash($k$)}\%$\Pi(\aggplusop)$ \;
}
\Class{$\zeta$}
{
    \FuncSty{long count}\;
}
\Proc{$f_U(w,t)$}
{
    $w.\zeta.$\FuncSty{count}$\xleftarrow{}w.\zeta.$\FuncSty{count}$+1$\;
    \Return \{$w.\zeta$\}\;
}
\Proc{$f_O(w,t)$}
{
    \Return $\{w.k,w.\zeta.\text{count}\}$\;
}
\caption{\draft{User-defined parameters for the $\aggplusop{}$ used in~\autoref{ssc:q1} for \xxx{}'s \texttt{wordcount}/\texttt{paircount} implementation.}}
\label{ex:stretch:wordpaircount}
\end{algorithm}

\begin{algorithm}[h]
\footnotesize
\SetAlgoLined
\DontPrintSemicolon
\SetKwInput{KwParameters}{Instance-local variables}
\SetKwInput{KwAuxiliaryFunctions}{Auxiliary functions}
\SetKwProg{Proc}{Function}{}{}
\SetKwProg{Class}{Class}{}{}

\BlankLine
\nonl $\generalopplus(\delta,\delta,2,\keybymulti,\text{single},S_O,\opmapping,f_U,-,-)$, where:\;  
\BlankLine

\Proc{$\keybymulti(t)$}
{
    \Return $\{1,\ldots,n\}$\;
}
\Proc{$\opmappingfun(k)$}
{
    \Return $k$\;
}
\Proc{$f_U(\{w_1,w_2\},t)$}
{
    \Return $\{\emptyset,\emptyset,t.\varphi\}$ \tcp{return empty states for $w_1$ and $w_2$ and $t$'s payload}
}
\caption{\generalopplusinmath{} for $Q_2$ and $\Pi(\generalopplus)=n$.}
\label{ex:o2}
\end{algorithm}

\newlength{\commentWidth}
\setlength{\commentWidth}{11.4cm}
\newcommand{\atcp}[1]{\tcp*[r]{\makebox[\commentWidth]{#1\hfill}}}

\renewcommand*{\algorithmcfname}{Execution Trace}

\begin{algorithm}[h]
\footnotesize
\setcounter{algocf}{0}
\SetAlgoLined
\DontPrintSemicolon

\SetKwInput{KwStart}{Initial State}

\KwStart{}
\nonl $W=\text{09:00}$\;
\nonl $\sigma=\{\langle 11,\text{09:00},``\text{pink}" \rangle,\langle 11,\text{09:30},``\text{pink}" \rangle\}$\;

\BlankLine
\FuncSty{update$W$}($t$) \tcp{$W\xleftarrow{}\text{09:58}$} \label{alg:exec:update} 
\FuncSty{handleInputTuple}($t$)\;\label{alg:exec:handleinputtuple} 
\FuncSty{earliestWinL}($t$) \tcp{$\tau_1 \xleftarrow{} \text{09:00}$}
\FuncSty{latestWinL}($t$) \tcp{$\tau_2 \xleftarrow{} \text{09:30}$} \label{alg:exec:latestWinL} 
\FuncSty{$\sigma_j$.check\&Create($``\text{red}",\text{09:00}$)}\; 
\tcp{$\sigma \xleftarrow{} \{\langle 11,\text{09:00},``\text{pink}" \rangle,\langle 0,\text{09:00},``\text{red}" \rangle$,\\$ \langle 11,\text{09:30},``\text{pink}" \rangle\}$} \label{alg:exec:updatestart}
\FuncSty{$\sigma_j$.set($``\text{red}",0,\{13\}$)}\; 
\tcp{$\sigma\xleftarrow{}\{\langle 11,\text{09:00},``\text{pink}" \rangle,\langle 13,\text{09:00},``\text{red}" \rangle$,\\$\langle 11,\text{09:30},``\text{pink}" \rangle\}$}
\FuncSty{$\sigma_j$.check\&Create($``\text{red}",\text{09:30}$)}\;
\tcp{$\sigma\xleftarrow{}\{\langle 11,\text{09:00},``\text{pink}" \rangle,\langle 13,\text{09:00},``\text{red}" \rangle$,\\$\langle 11,\text{09:30},``\text{pink}" \rangle,\langle 0,\text{09:30},``\text{red}" \rangle\}$}
\FuncSty{$\sigma_j$.set($``\text{red}",1,\{13\}$)}\;
\tcp{$\sigma\xleftarrow{}\{\langle 11,\text{09:00},``\text{pink}" \rangle,\langle 13,\text{09:00},``\text{red}" \rangle$,\\$\langle 11,\text{09:30},``\text{pink}" \rangle,\langle 13,\text{09:30},``\text{red}" \rangle\}$}
\FuncSty{$\sigma_j$.check\&Create($``\text{pink}",\text{09:00}$)}\;
\FuncSty{$\sigma_j$.set($``\text{pink}",0,\{13\}$)}\;
\tcp{$\sigma\xleftarrow{}\{\langle 13,\text{09:00},``\text{pink}" \rangle,\langle 13,\text{09:00},``\text{red}" \rangle$,\\$\langle 11,\text{09:30},``\text{pink}" \rangle,\langle 13,\text{09:30},``\text{red}" \rangle\}$}
\FuncSty{$\sigma_j$.check\&Create($``\text{pink}",\text{09:30}$)}\;
\FuncSty{$\sigma_j$.set($``\text{pink}",1,\{13\}$)}\;
\tcp{$\sigma\xleftarrow{}\{\langle 13,\text{09:00},``\text{pink}" \rangle,\langle 13,\text{09:00},``\text{red}" \rangle$,\\$\langle 13,\text{09:30},``\text{pink}" \rangle,\langle 13,\text{09:30},``\text{red}" \rangle\}$} \label{alg:exec:updateend}

\caption{Execution trace for the $\aggplusop$ in Appendix~\ref{apx:examples}, upon invocation of $\bm{\texttt{processSN}(t=\langle \text{09:58}, \left[\text{C}, ``\text{hi \#red \#pink}"\right] \rangle})$}
\label{alg:exec}
\end{algorithm}

\section{Example for Theorem~\ref{thm:generalizationworks}}
\label{apx:appendixexecution}

\draft{
In here, we continue the example introduced in Appendix~\ref{apx:appendixcorollary}, focusing on the execution path (in terms of invoked functions) that is triggered for the $\aggplusop$ of Appendix~\ref{apx:examples} (which implements the same semantics of the $\mapop{}$ and $\aggop{}$ from \autoref{fig:corollary1example}) upon the reception of tuple $t=\langle \text{09:58}, \left[\text{C}, ``\text{hi \#red \#pink}"\right] \rangle$.
We assume in this case that the \WA{} and \WS{} parameters for $\aggplusop$ are set to 30 minutes and 1 hour, respectively.
As we show next, the events triggered on $\generalopplus$'s state upon reception of $t$ are the same triggered for $\aggop{}$ in \autoref{fig:corollary1example} upon the reception of the two tuples produced by $\mapop{}$ from $t$, i.e. $\langle \text{09:58}, \left[\text{red}, 13\right] \rangle$ and $\langle \text{09:58}, \left[\text{pink},13\right] \rangle$.}

\draft{
For simplicity, we assume all the input tuples' timestamps from \autoref{fig:corollary1example} are valid watermarks for $\aggplusop$ (as stated in \autoref{sec:smps} such information can be carried in tuples' metadata). Hence, given the \WA{}/\WS{} parameters of $\aggop{}$/$\aggplusop$, and since the tuple previously processed by $\aggop{}$/$\aggplusop$ was $\langle \text{09:50}, \left[\text{B}, ``\text{hello \#pink}"\right] \rangle$, only window instances such that $l \geq 09:00$ are maintained by $\aggplusop$ upon reception of $t$. 
}

\draft{
The initial state and the execution path following the invocation of \texttt{processSN} for tuple $t$ are shown in Listing~\ref{alg:exec}. 
The first invoked method is \texttt{update$W$}, and $W$ accordingly updated to 09:58 (L\ref{alg:exec:update}).
Since $\psipar{}=\text{multi}$, $\tau_1$ and $\tau_2$ are set to 09:00 and 09:30 within the execution of method \texttt{handleInpuTuple} (L\ref{alg:exec:handleinputtuple}-\ref{alg:exec:latestWinL}).
Subsequently, all window instances for key ``red" and ``pink", and for $l=\text{09:00}$ and $l=\text{09:30}$, are updated (L\ref{alg:exec:updatestart}-\ref{alg:exec:updateend}). Since the window instances for key ``pink" are already in place before the processing of $t$, only the counter maintained in their $\zeta$ is updated. As shown in L\ref{alg:exec:updateend}, the final state maintained by $\aggplusop$, contains the window instances that will then result in the production of tuples $\langle \text{10:00}, \left[\text{red}, 13\right] \rangle$ and $\langle \text{10:00}, \left[\text{pink},13\right] \rangle$ (\autoref{fig:corollary1example}) once their timestamp is set to $l+\WS{}$ (i.e., to 10:00).
}
\section{Additional Experiments}
\label{apx:addexperiments}

\draft{
This section contains additional experiments, similar to the one presented in \autoref{ssc:q5}. These 20 minutes long experiments stress-test \xxx{}’s reconfigurations with several phases in which the input rate is subject to abrupt changes.
}

\begin{figure}[h]
  \centering
  \includegraphics[width=0.8\linewidth]{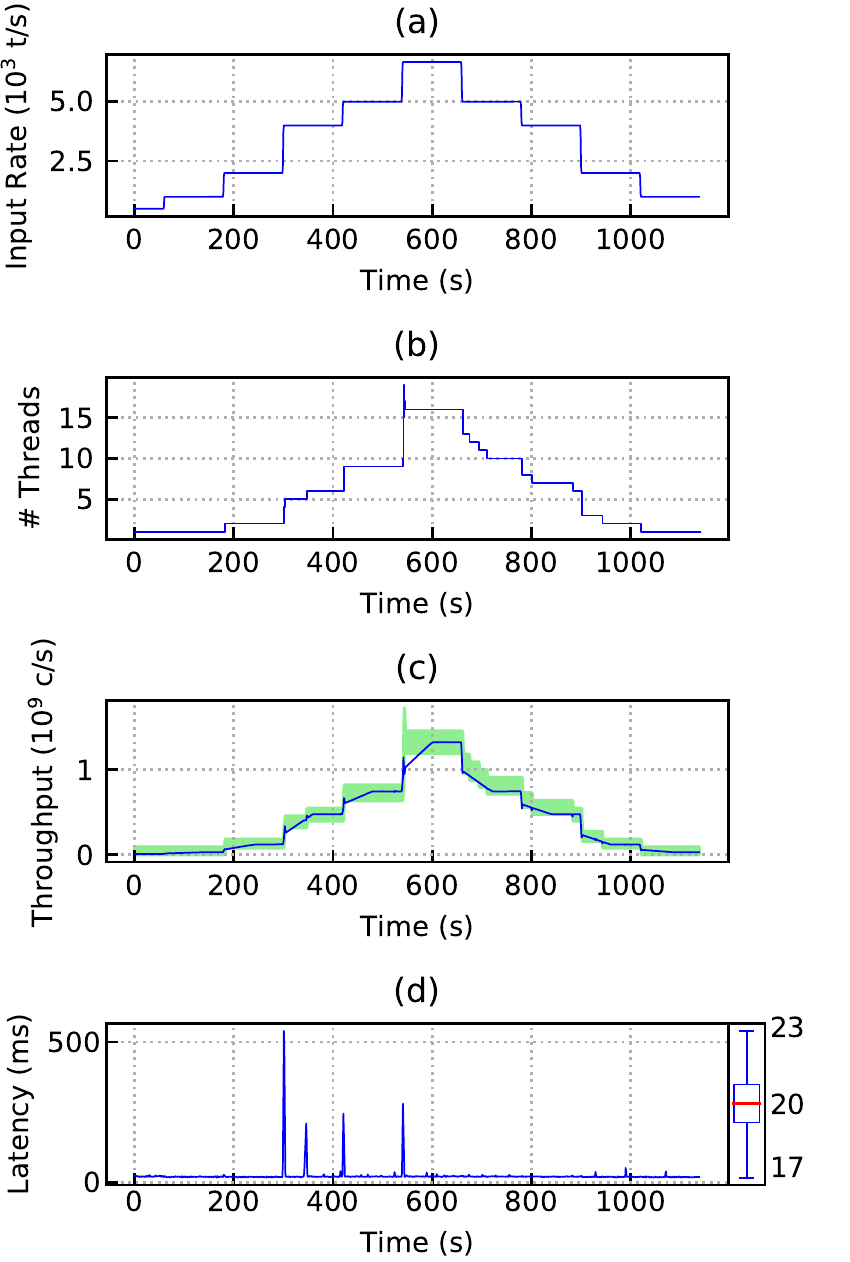}
  \caption{Results of adjusting the number of processing threads with respect to the input rate for synthetic data.}
  \label{fig:rateSeq2}
\end{figure}

\begin{figure}[h]
  \centering
  \includegraphics[width=0.8\linewidth]{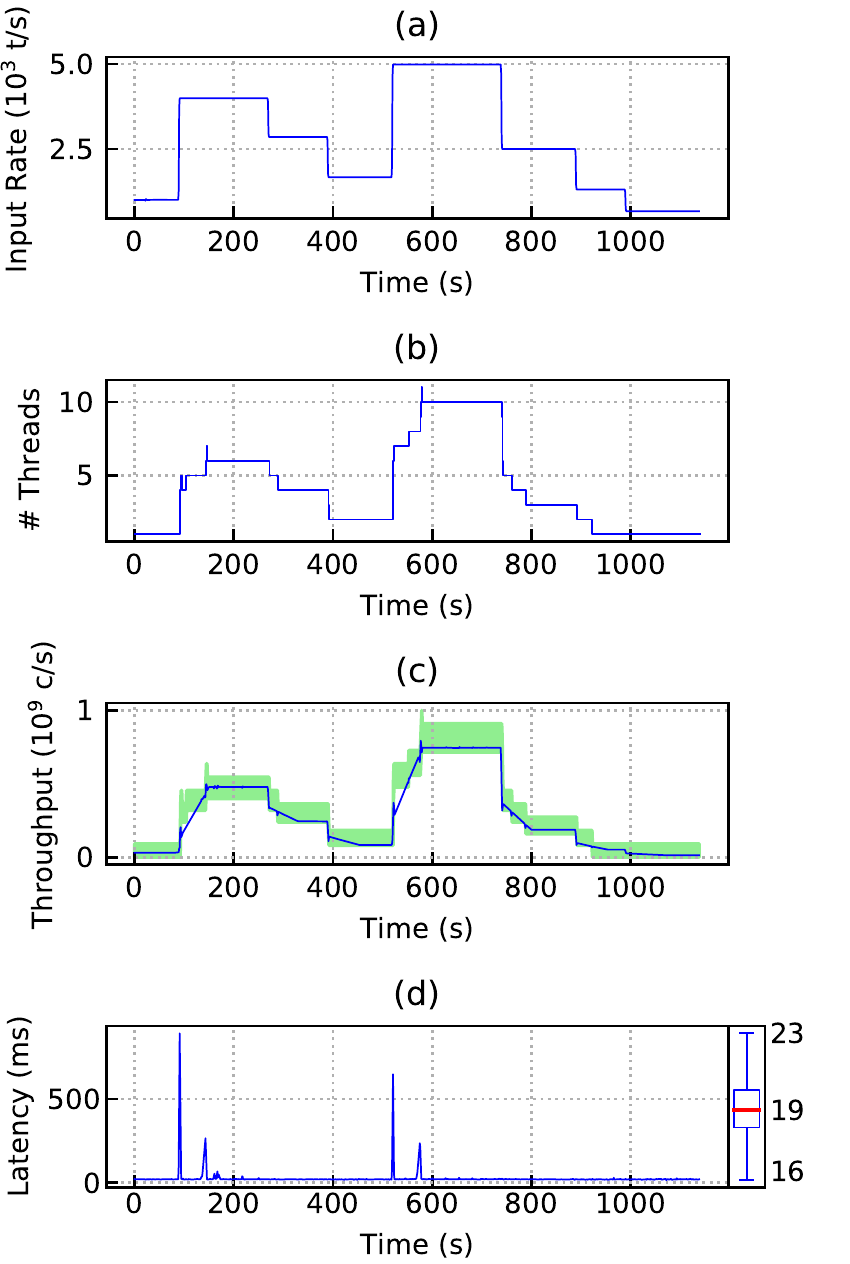}
  \caption{Results of adjusting the number of processing threads with respect to the input rate for synthetic data.}
  \label{fig:rateSeq3}
\end{figure}

\begin{figure}[h]
  \centering
  \includegraphics[width=0.8\linewidth]{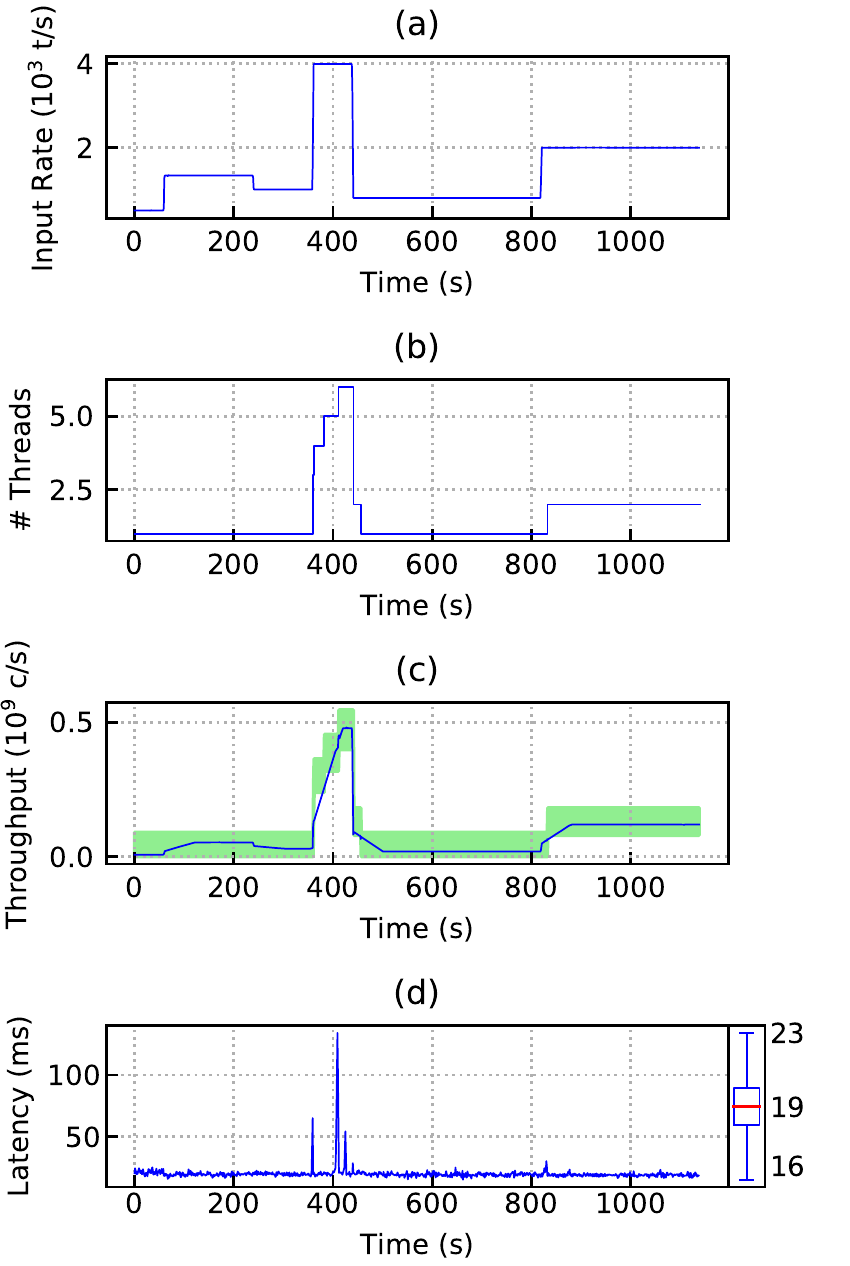}
  \caption{Results of adjusting the number of processing threads with respect to the input rate for synthetic data.}
  \label{fig:rateSeq4}
\end{figure}

\begin{figure}[h]
  \centering
  \includegraphics[width=0.8\linewidth]{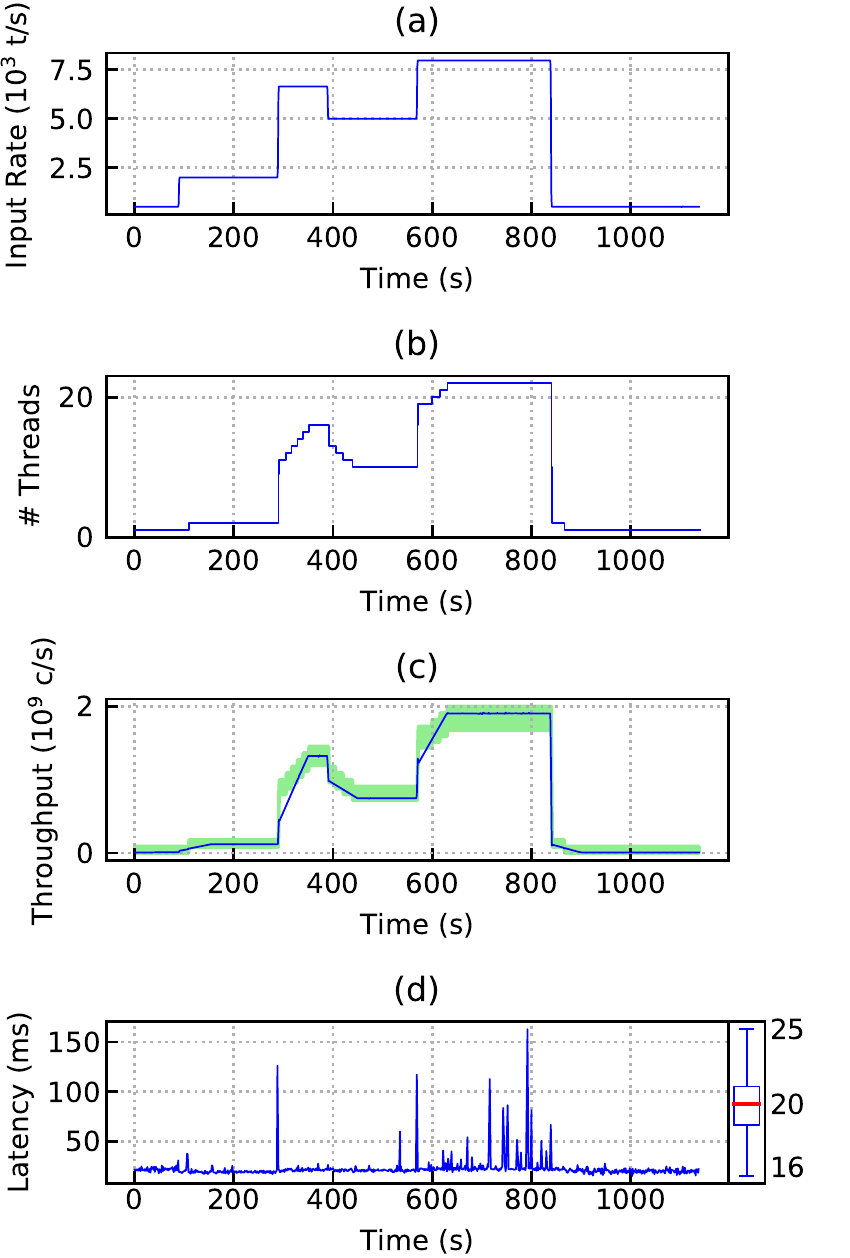}
  \caption{Results of adjusting the number of processing threads with respect to the input rate for synthetic data.}
  \label{fig:rateSeq5}
\end{figure}


\end{document}